\documentclass[sigplan,screen]{acmart} 

\usepackage[T1]{fontenc}
\usepackage[utf8]{inputenc}

\usepackage{cleveref}

\usepackage{listings, fancyvrb}
\usepackage{amsmath, amsthm, xspace}

\usepackage{eqparbox} 
\usepackage{algorithm}
\usepackage[noend]{algpseudocode}  

\usepackage{enumitem} 
\usepackage{nicefrac}       
\usepackage{graphicx}
\usepackage[skip=0pt]{subcaption}

\usepackage{hyperref}
\usepackage{multirow}

\renewcommand\footnotetextcopyrightpermission[1]{}

\usepackage[framemethod=TikZ]{mdframed}
\mdfdefinestyle{listingstyle}{%
  outerlinewidth=0.25pt,outerlinecolor=black,%
  innerleftmargin=5pt,innerrightmargin=5pt,innertopmargin=0pt,innerbottommargin=0pt%
}

\AtBeginDocument{%
	\providecommand\BibTeX{{%
			\normalfont B\kern-0.5em{\scshape i\kern-0.25em b}\kern-0.8em\TeX}}}

\usepackage{etoolbox}

\newcommand{\algstrut}[1][\algruledefaultfactor]{\vrule width 0pt
depth .25\baselineskip height #1\baselineskip\relax}
\newcommand*{\algrule}[1][\algorithmicindent]{\hspace*{.5em}\vrule\algstrut
\hspace*{\dimexpr#1-.5em}}

\makeatletter
\newcount\ALG@printindent@tempcnta
\def\ALG@printindent{%
    \ifnum \theALG@nested>0
    \ifx\ALG@text\ALG@x@notext
    \else
    \unskip
    \ALG@printindent@tempcnta=1
    \loop
    \algrule[\csname ALG@ind@\the\ALG@printindent@tempcnta\endcsname]%
    \advance \ALG@printindent@tempcnta 1
    \ifnum \ALG@printindent@tempcnta<\numexpr\theALG@nested+1\relax
    \repeat
    \fi
    \fi
}%

\patchcmd{\ALG@doentity}{\noindent\hskip\ALG@tlm}{\ALG@printindent}{}{\errmessage{failed to patch}}

\AtBeginEnvironment{algorithmic}{\lineskip0pt}

\algnewcommand\StartFromLine[1]{\setcounter{ALG@line}{\numexpr#1-1}}

\ifthenelse{\equal{\ALG@noend}{t}}{%
	\algdef{SE}[FUNCTION]{Function}{EndFunction}[3]{\op{#2}\ifthenelse{\equal{#3}{}}{}{(#3)}\ifthenelse{\equal{#1}{}}{}{ : #1}}{}%
	\algtext*{EndFunction}%
	}{
	\algdef{SE}[FUNCTION]{Function}{EndFunction}[3]{{\op{#2}}\ifthenelse{\equal{#3}{}}{}{(#3)}\ifthenelse{\equal{#1}{}}{}{ : #1}}[1]{\algorithmicend\ \op{#1}}%
	}
	
\algblockdefx{Rpt}{EndRpt}[1]{\algorithmicrepeat\ #1}{\algorithmicend\ \algorithmicrepeat}

\newcommand{\algorithmicclass}{\textbf{Class}}

\ifthenelse{\equal{\ALG@noend}{t}}{%
	\algdef{SE}[CLASS]{Class}{EndClass}[1]{\algorithmicclass\ #1}{}%
	\algtext*{EndClass}%
	}{
	\algdef{SE}[CLASS]{Class}{EndClass}[1]{\algorithmic class\ #1}[1]{\algorithmicend\ \op{#1}}%
	}

\makeatletter
\renewcommand{\ALG@beginalgorithmic}{\small}
\makeatother

\makeatletter
\patchcmd{\ALG@doentity}{\item[]\nointerlineskip}{}{}{}
\makeatother

\def\draft{0}  
\if\draft 1
\newcommand{\mycomment}[3]{{\color{#2}{[\textbf{#1: #3}]}}}
\newcommand{\myedit}[2]{{\color{#1}{#2}}\normalcolor}
\newcommand{\here}[1]{\textbf{[[[#1]]]}}
\newcommand{\future}[1]{[[To think about in future: #1]]}
\else
\newcommand{\mycomment}[3]{}
\newcommand{\myedit}[2]{#2}
\newcommand{\here}[1]{}
\newcommand{\future}[1]{} 
\fi

\newcommand{\Eric}[1]{\mycomment{Eric}{blue}{#1}}
\newcommand{\Siddhartha}[1]{\mycomment{Siddhartha}{green}{#1}}

\newcommand{\Y}[1]{\mycomment{Youla}{brown}{#1 }}

\newcommand{\y}[1]{\myedit{brown}{#1 }}

\newcommand{\e}[1]{\myedit{blue}{#1 }}
\newcommand{\remove}[1]{}

\newcommand{\x}[1]{\mbox{\textit{#1}}} 
\newcommand{\op}[1]{\mbox{\textsc{#1}}} 
\newcommand{\FetchAndAdd}[0]{\op{Fetch\&Add}}
\newcommand{\FetchAndAddDirect}[0]{\op{Fetch\&AddDirect}}
\newcommand{\Read}[0]{\op{Read}}
\newcommand{\CompareAndSwap}[0]{\op{Compare\&Swap}}

\newcommand{\FAA}[0]{\op{F\&A}}
\newcommand{\CAS}[0]{\op{CAS}}

\newcommand{\Agg}{\mbox{Aggregator}}
\newcommand{\Aggs}{\mbox{\Agg s}}

\newcommand{\floor}[1]{\left\lfloor #1 \right\rfloor}

\makeatletter
\lst@Key{countblanklines}{true}[t]%
{\lstKV@SetIf{#1}\lst@ifcountblanklines}

\lst@AddToHook{OnEmptyLine}{%
	\lst@ifnumberblanklines\else%
	\lst@ifcountblanklines\else%
	\advance\c@lstnumber-\@ne\relax%
	\fi%
	\fi}
\lst@AddToHook{OnEmptyLine}{\vspace{-0.3\baselineskip}}
\makeatother

\newtheorem{theorem}{Theorem}[section]
\newtheorem{lemma}[theorem]{Lemma}

\newtheorem{invariant}[theorem]{Invariant}

 \author{Younghun Roh}
 \affiliation{%
  	\institution{MIT CSAIL}
 	\country{USA}
 }
 \email{yhunroh@mit.edu}

  \author{Yuanhao Wei}
 \affiliation{%
    \institution{MIT CSAIL}
 	\country{USA}
 }
 \email{yuanhao1@mit.edu}

  \author{Eric Ruppert}
 \affiliation{%
    \institution{York University}
 	\country{Canada}
 }
 \email{eruppert@yorku.ca}

 \author{Panagiota Fatourou}
 \affiliation{%
    \institution{FORTH ICS and University of Crete}
    \country{Greece}
 }
 \email{faturu@cs.uoc.gr}

 \author{Siddhartha Jayanti}
 \affiliation{%
    \institution{Google Research}
 	\country{USA}
 }
 \email{sjayanti@google.com}

 \author{Julian Shun}
 \affiliation{%
    \institution{MIT CSAIL}
 	\country{USA}
 }
 \email{jshun@mit.edu}



\begin{document}

\copyrightyear{2025}
\acmYear{2025}
\setcopyright{cc}
\setcctype{by}
\acmConference[PPoPP '25]{The 30th ACM SIGPLAN Annual Symposium on
Principles and Practice of Parallel Programming}{March 1--5, 2025}{Las Vegas,
NV, USA}
\acmBooktitle{The 30th ACM SIGPLAN Annual Symposium on Principles and
Practice of Parallel Programming (PPoPP '25), March 1--5, 2025, Las Vegas, NV,
USA}
\acmDOI{10.1145/3710848.3710873}
\acmISBN{979-8-4007-1443-6/25/03}

\title{
Aggregating Funnels for Faster Fetch\&Add and Queues
}


\begin{abstract} 
Many concurrent algorithms require processes to perform fetch-and-add operations on a single memory location, which can be a hot spot of contention.
We present a novel algorithm called {\em Aggregating Funnels} that reduces this contention by spreading the fetch-and-add operations across multiple memory locations. 
It aggregates fetch-and-add operations into batches so that the batch can be performed by a single hardware fetch-and-add instruction on one location and all operations in the batch can efficiently compute their results by performing a fetch-and-add instruction on a \emph{different} location.
We show experimentally that this approach achieves higher throughput than previous combining techniques, such as Combining Funnels,
and is substantially more scalable than applying hardware fetch-and-add instructions on a single memory location.
We show that replacing the fetch-and-add instructions in the fastest state-of-the-art concurrent queue by 
our Aggregating Funnels
eliminates a bottleneck and greatly improves the queue's overall throughput. 
\end{abstract}

	\maketitle 
	\section{Introduction}

Many concurrent algorithms use fetch-and-add  
to coordinate the actions of multiple processes. 
A \emph{fetch-and-add} on a memory location $X$ {\em atomically} adds a given value to $X$ and returns the value that was stored in $X$ before the addition.
Introduced by Gottlieb and Kruskal~\cite{GK81}, fetch-and-add is widely available as a hardware {\em primitive} \cite{X86}.
Applications often have hot spots of contention where many processes perform concurrent fetch-and-adds on the same location, degrading performance.  To mitigate this problem, we introduce Aggregating Funnels,  a software implementation of fetch-and-add that is much more scalable than the hardware primitive, and more efficient than state-of-the-art software implementations.
Throughout the paper, we use \FetchAndAdd\ to denote  software implementations and \FAA\ for the hardware primitive.
Since our implementation is linearizable \cite{HW90}, our \FetchAndAdd\ can be used in place
of \FAA\ in any application.

Scalable and efficient software replacements for hardware \FAA\ are crucial for obtaining high performance in many concurrent algorithms. 
Applications of \FAA\ include allocating memory addresses for objects of varying size 
\cite{Sto82,EO88,Wil88},
solving the readers-writers problem \cite{GLR83}, and
wait-free universal constructions~\cite{FK14}.
\FAA\ can be used to implement simpler primitives, such as 
\op{Fetch\&Inc}---which simply amounts to performing \FAA(1)---and {\em counters}, which
support the operations \op{Add}(\x{val})  and \op{Read} (via \FAA(\x{val}) and \FAA(0), respectively).
These primitives themselves have a plethora of applications.
For example, \op{Fetch\&Inc} is used in
assigning distinct identifiers to processes \cite{MA95},
reference counting for garbage collection in concurrent systems \cite{Val95,MS95,Sun05}, 
assigning distinct tickets in a ticket lock \cite{FLBB79,RK79,MS91},
assigning distinct timestamps to operations \cite{BJW23},  
implementing simple barriers~\cite[Chapter 18.2--18.3]{HSLS21}, 
array-based queue   locks~\cite{A90} (see also~\cite[Chapter 7.5.1]{HSLS21}),
highly-efficient concurrent data structures, such as queues~\cite{GLR83,FG91,MA13,FK14,R19,FGM24} and stacks~\cite{PZ18},
and in many other applications~\cite{FM+90,FK83,KMZ84,Pet82}.

Previous work ~\cite{HSLS21, FK12, SZ00} provided
implementations of \FetchAndAdd\ that alleviate the bottleneck of multiple processes simultaneously performing \FetchAndAdd\  on a single memory location. They use \emph{software combining} to diffuse the contention.
Active \FetchAndAdd\ operations coordinate on low-contention ancillary variables, combine their operations, and choose a delegate, 
so that only the delegates contend for the main variable.
The delegate then reports its return value to the \FetchAndAdd\ operations waiting on it and the waiting operations use this to calculate their own return values.
This combining process ensures that both the ancillary and main variables have low contention.
This line of work culminated in a technique called \emph{Combining Funnels}~\cite{SZ00}, which filters \FetchAndAdd\ operations through several levels of objects called funnels, 
combining them pair-wise at each funnel using swap and compare-and-swap primitives.

While these existing approaches  reduce contention on each variable access, they considerably increase the number of variables accessed by a \FetchAndAdd.
The additional cost of these accesses outweighs the benefits when there are few concurrent \FetchAndAdd\ operations.
Indeed, we see in our experiments that using Combining Funnels is significantly slower than hardware \FAA\ on low thread counts and only slightly outperforms  \FAA\ after 100 threads in \FetchAndAdd-heavy workloads (see ~\Cref{fig:count_thr,fig:count_100} in \Cref{sec:experiments}).

We present \emph{Aggregating Funnels}, a novel way to implement \FetchAndAdd\ that 
significantly reduces contention  while only slightly increasing the number of variables accessed.
\e{We also use multiple levels, with $k$-way combining of \FetchAndAdd\ operations at each level that is efficient for up to $k=25$ in our experiments.}
This lets us combine more at each level, and use far fewer levels than the Combining Funnels approach. 
In fact, using just \emph{one} level of Aggregating Funnels yielded the best performance in our experiments.

A technique for $k$-way combining was previously proposed by Tang and Yew~\cite{TY90}, but their algorithm is more complex and uses primitives not available on modern machines, such as \op{Fetch\&Add\&Store}, which atomically performs \FetchAndAdd\ on one memory location and \op{Store} on an adjacent memory location.
In contrast, our algorithm uses only the widely available \op{Load}, \op{Store}, and \FAA\ primitives.


We call the mechanism used to achieve $k$-way combining in our algorithm  an \emph{aggregator}.
Aggregators achieve fast combining by having each thread register itself in a batch using a single \FAA\ instruction.
This \FAA\ contends only with other threads accessing the same aggregator and it serves multiple purposes.
It is used to (1) decide the delegate for each batch, (2) sum all of the operations within the batch, (3) determine when the batch is closed, and (4) help determine the return value of the \FetchAndAdd.
Previous combining techniques~\cite{HSLS21, FK12, SZ00} use several variables to coordinate these tasks.
Accomplishing all these tasks with a single \FAA\ per operation is one reason for the increased efficiency of our combining approach.
Our experiments show that Aggregating Funnels start outperforming hardware \FAA\ for as few as 30 threads and are up to 4x faster than both Combining Funnels and hardware \FAA\ at high thread counts.

We prove that our \FetchAndAdd\ implementation is \emph{strongly linearizable} \cite{GHW11}, 
making it suitable for deployment even in randomized concurrent algorithms.
It is \emph{blocking} because combined operations must wait for the delegate to bring back a return value from the main variable, but our experiments show  this does not lead to uneven performance of threads.  In fact, our
implementation provides greater fairness (i.e., it results in more similar throughputs at different threads) than using hardware \FAA\ directly.
We also provide a \op{Fetch\&AddDirect} operation which can be used by higher priority threads to skip the combining step and go directly to the main variable.
Our experiments show that this direct option significantly increases the \FetchAndAdd\ throughput of the high priority threads without reducing the overall throughput.
Our implementation is \emph{RMWable} \cite{JJJ24}, meaning that it 
also supports any other operation that is provided as an atomic primitive.  
For example, if the hardware provides a compare-and-swap instruction,
then a \op{Compare\&Swap} can also be 
supported by our fetch-and-add object.

\noindent{\bf Scaling up concurrent queues.} We believe plugging Aggregating Funnels into many applications that use \FAA\ 
will make them more scalable.
As evidence, we use them to implement the highly contended fetch-and-add objects in the concurrent queue LCRQ, published in PPoPP 2013~\cite{MA13} and recently shown to still be the
fastest concurrent queue~\cite{RK23}.
Aggregating Funnels eliminate the scalability bottleneck in LCRQ,  improving throughput by up to 2.5x for high thread counts.
This is a significant leap forward in concurrent queue efficiency.
This speed-up also highlights how the significant efficiency
gains observed in the microbenchmarks indeed translate into
better performing higher-level applications.

\noindent
{\bf Our Contributions.} We summarize our main contributions.
\begin{itemize}[left= 0pt, labelwidth=7pt, itemindent=11pt,topsep=0pt,partopsep=0pt]
    \item 
    We design the {\em Aggregating Funnels} algorithm, which uses hardware \FAA\ instructions to implement \FetchAndAdd\ operations with greatly reduced contention on individual memory locations.  It also provides greater fairness to threads.

    \item
    We 
    show that Aggregating Funnels provide more scalable \FetchAndAdd\ operations than hardware \FAA\ and the state-of-the-art Combining Funnel algorithm \cite{SZ00}.

    \item        
    Our experiments show that replacing hardware \FAA\ with our Aggregating Funnels makes the fastest available concurrent queue, LCRQ \cite{MA13}, significantly faster and more scalable.

    \item 
    Our implementation is \emph{strongly linearizable} \cite{GHW11}, making it suitable as a replacement for hardware \FetchAndAdd\  in both deterministic and randomized algorithms.
    
    \item 
    Our implementation supports all hardware primitives.
\end{itemize}

	\section{Related Work}
\label{sec:related-work}

\remove{
A concurrent \emph{counter} $V$ stores an integer and supports atomic operations \op{Add}($v$), which adds $v$ to $V$ and returns ack, 
and \op{Read}() which returns the current value of $V$. 
\FetchAndAdd trivially implements a counter using \FetchAndAdd(0) for \op{Read}. 
Similarly, \FetchAndAdd\ implements a timestamp object supporting  operations \op{GetTimestamp}(), 
which returns a unique timestamp, and \op{Compare}($\x{ts}, \x{ts}â€™$), which returns \textsc{True} 
if \x{ts} was obtained earlier than $\x{ts}â€™$ and \textsc{False} otherwise. 
\FetchAndAdd, with consensus number 2~\cite{H91}, is stronger than counters and timestamps, which have consensus number 1. 
}

Many papers have focused on designing practical implementations of \FetchAndAdd. 
%
One approach uses a complete tree of height $\Theta(\log p)$,
assigning a leaf to each of the $p$ threads. 
Each tree node stores some metadata (including a counter). 
To execute a \FetchAndAdd($\textit{diff}$), a thread first increments the counter of its leaf by $\textit{diff}$. 
Then, it works its way up the tree to the root, updating all nodes of the path it traverses.
%
Combining trees~\cite[Section 12.2]{HSLS21} (originally proposed in~\cite{GVW89,YTL87}), 
employ combining at each node of the tree to reduce contention. 
Every tree node contains a lock. Threads compete for this lock to ascend from a node to its parent, 
and only the winning thread proceeds to the next level up the tree.
The other thread waits for the winner to apply its operation. 

Combining trees were criticized as having performance that is sensitive to changes in the arrival 
rate of operations~\cite{HBS95,SZ96,SZ00}: whenever only a subset of the threads is concurrently active, 
little combining occurs, yet threads must still pay the cost of going through all $\Theta(\log p)$ levels
to reach the root. Combining Funnels~\cite{SZ00} address this by replacing the static tree with a 
series of combining layers, through which \FetchAndAdd\ operations are passed. 
In each layer, threads meet for combining by randomly choosing a location in an array at which to wait for other threads.
By using an adaptive scheme, the funnel can change its width and depth to accommodate dynamic access patterns. 
Combining funnels have been experimentally shown~\cite{SZ00} to outperform all schemes discussed in this section
that aim to provide low-contention implementations of \FetchAndAdd.

Counting networks~\cite{AHS94,M00} can be used to count concurrently and asynchronously.
They are constructed from simple two-input, two-output computing elements called {\em balancers}, 
connected to one another by wires. Threads traverse different paths through the network and obtain 
a return value when they reach an output wire of the network. 
Generalizations of counting networks can be used to implement \FetchAndAdd~\cite{FH04}.
A drawback of 
this approach is that 
linearizability cannot be supported without paying a significant cost or sacrificing other desirable properties~\cite{HSW96, FH04}. 
Diffracting trees~\cite{SZ96} employ some form of tree-like counting networks and attempt to reduce contention using arrays ({\em prisms}) where threads attempt to meet and combine. They achieve better performance than the counting networks in~\cite{AHS94}, 
but  are not linearizable.  
Counting networks and diffracting trees, even when they are used to implement \op{Fetch\&Inc},  are less efficient than combining funnels \cite{SZ00}.

Other work aims to efficiently implement objects (including \FetchAndAdd) from \emph{other} primitives,
such as \CAS\ and LL/SC.
Jayanti \cite{Jay02} used a technique known as {\em double refresh} (originally proposed in~\cite{ADT95}), to build $f$-arrays, where $f$ is a fixed function 
over the elements of an array $A[1,\ldots,n]$; each thread $i$ can update $A[i]$ and query the value of $f$. 
An $f$-array can be used to obtain a wait-free implementation of a concurrent counter (supporting \op{Add} and \op{Read}), whose step complexity is $O(\log p)$. 
Ellen and Woelfel~\cite{EW13} presented a wait-free, linearizable implementation of \FetchAndAdd\ with $O(\log p)$ worst-case step complexity from $O(\log m)$-bit registers and LL/SC objects for up to $m$ \FetchAndAdd\ operations.
These papers 
use the standard measure of step complexity, which simply counts the maximum number of accesses to shared variables that
an operation performs, without taking into account the contention that these accesses may cause.

Jayanti~\cite{J98} proved an $\Omega(\log p)$ lower bound on the expected step complexity of any randomized wait-free, linearizable implementation of a single-shot \op{Fetch\&Inc} object from \op{LL/SC} objects (for a strong adaptive adversary).  
This bound also holds for long-lived objects, even with amortization \cite{JTB19, Jay22}.
Randomized implementations of \op{Fetch\&Inc} with polylogarithmic step complexity are known~\cite{AA+11}. 

\remove{
Ellen et al. (theoretical) algorithm for $O(\log p)$ implementation of F\&Inc using LL/SC: \url{https://link.springer.com/chapter/10.1007/978-3-642-33651-5_2}
Chandra et al. (theoretical) algorithm for $O(\log^2 p)$ implementation of F\&A using LL/SC: \url{https://dl.acm.org/doi/pdf/10.1145/277697.277753}
Also look at \cite{TY90}}

\noindent
{\bf Concurrent Queues.}
\label{related-queue}
State-of-the-art concurrent queues employ \FAA~\cite{MA13,R19,RK23}. 
There is empirical evidence \cite{RK23} that 
LCRQ~\cite{MA13} has the best performance. 
LCRQ is inspired by the simple idea of using an infinite array, $Q$, 
and two indices \x{Tail} and \x{Head} that are updated using \op{Fetch\&Inc}.
Initially, all elements of $Q$
contain $\bot$, and $\x{Head}=\x{Tail}=0$. 
\e{To enqueue an item, a thread repeatedly performs a \op{Fetch\&Inc} on \x{Tail} to get an index $i$ and swaps the item into $Q[i]$
using a \op{Fetch\&Store} until it replaces a $\bot$ with its item. 
A dequeue repeatedly executes \op{Fetch\&Inc} on \x{Head} to obtain an index $i$
and tries swap $\top$ into $Q[i]$ until one such swap returns a non-$\bot$ item, which it can return 
(or until it detects that the queue is empty). 
This way, each element of $Q$
is accessed by at most one enqueue and one dequeue. 
An enqueue whose swap into $Q[i]$ returns $\top$
knows that a dequeue has already accessed $Q[i]$, so the enqueue continues trying to swap its item into other locations.}
To bound space usage, LCRQ uses a linked list of circular arrays in place of the infinite array $Q$.

LCRQ is lock-free but uses double-word \CAS. 
LPRQ~\cite{RK23} is a variant of LCRQ that uses single-word \CAS, but it does not outperform
LCRQ in empirical tests.
Recent work~\cite{YMC16,NR22} provides wait-free 
concurrent queues based on \FAA, but they
also do not outperform LCRQ in experimental analyses.

        \section{Aggregating Funnels Algorithm}
\label{sec:algorithms}

We present a linearizable implementation of a fetch-and-add object $O$ that stores an integer and supports the operations
\FetchAndAdd(\x{df}), which adds \x{df} to the value of $O$ and returns its previous value,
and \Read, which returns the value of $O$.
The implementation uses atomic \op{Read}, \op{Write} and \FAA\ as primitives (i.e., hardware instructions).
Any other 
primitives, such as compare-and-swap, that are supported
by hardware are also supported by $O$.

Our implementation uses a principal variable \x{Main}, which stores the actual value of $O$,
and $2m$ ancillary objects called \Aggs, which aggregate batches of concurrent\linebreak \FetchAndAdd\ operations.
One operation from each batch is chosen 
to apply a single \FAA(\x{sum}) on \x{Main}, where \x{sum}
is the sum of the batch's arguments. That operation is called the {\em delegate} of the batch.
Thus, an \Agg\ acts as a funnel to narrow a stream of operations:  many operations may arrive at the \Agg\ concurrently, but only one operation at a time proceeds to access \x{Main}.
The goal 
is to limit contention by spreading out \FAA\ primitives
across $2m+1$ memory locations (\x{Main} and $2m$ \Aggs) 
instead of having all operations 
perform a \FAA\ on the same location.
\Eric{We sometimes used ``a \FAA'' and sometimes ``an \FAA''.  In my head, I pronounce \FAA\ as ``fetch and add'' rather than ``eff and a'', so I changed to ``a \FAA''.}

A \FetchAndAdd(\x{df}) operation first chooses one of the $2m$ Aggregators:
$m$ \Aggs\ are used for \FetchAndAdd\ operations with positive arguments, and the other
$m$ are used for negative arguments (we will later discuss possible ways to choose an Aggregator).
It applies a \FAA(\x{df}) to the \x{value} field of its chosen \Agg.
Thus, \x{value} stores the sum of the arguments of \FetchAndAdd\ operations that have been applied to the \Agg.
Each \Agg\ also stores additional information, described below,
to help operations on $O$ compute the results that they should return.
Our implementation works for any value of $m$.  
Thus, in practice, we choose a value of $m$ to optimize performance.

\subsection{Detailed Description}

\colorlet{myblue}{cyan!60!black} 
\newcommand{\threshold}{\x{Threshold}}

\begin{algorithm*}
\begin{algorithmic}[1]
\Class{\Agg}\Comment{used to aggregate batches of \FetchAndAdd\ operations}
	\State unsigned int64 \x{value} \Comment{sum of values added at this \Agg}
	\State Batch* \x{last} \Comment{last Batch in \Agg's list}%
\begingroup \color{myblue}
	\State unsigned int64 \x{final} \Comment{value after final batch, or $\infty$ if \Agg\ is still in use}%
\endgroup
\EndClass	

\medskip

\Class{Batch} \Comment{represents a batch of operations on an \Agg}
	\State unsigned int64 \x{before} \Comment{\Agg's \x{value} before the batch}
	\State unsigned int64 \x{after} \Comment{\Agg's \x{value} after the batch}
	\State unsigned int64 \x{mainBefore} \Comment{value of \x{Main} before the batch}
	\State Batch* \x{previous} \Comment{pointer to previous Batch of operations on the \Agg}
\EndClass

\medskip

\State Shared variables:
\State int $\x{Main} \leftarrow 0$
\State \Agg* \x{Agg}$[0,\ldots,2m-1]$ \Comment{first $m$ for positive arguments and the rest for negative ones}%
\begingroup\color{myblue} \State unsigned int64 $\threshold \leftarrow 2^{63}$ \label{rep-max} \Comment{when an \Agg's \x{value} exceeds this, it is retired}%
\endgroup

\For{$i\leftarrow 0,\ldots,2m-1$} \Comment{initialize \x{Agg} array}
    \State $\x{Agg}[i]\leftarrow \mbox{new \Agg}(0,\mbox{new Batch}(0,0,0,\perp)\begingroup\color{myblue},\infty\endgroup)$
\EndFor
\medskip

\Function{int}{Read(\hphantom{$\cdot$})}{} \Comment{read value directly from \x{Main}}
	\State \Return{\x{Main}}
\EndFunction

\medskip

\Function{int}{Fetch\&Add}{int \x{df}}
        \If{$\x{df} = 0$}
            \Return \op{Read}()\label{FAA0}
        \EndIf
        \State int $index \leftarrow $ \op{Choose\Agg}(\x{df}) \Comment{\x{index} should be in $0,\ldots,m-1$ if and only if $\x{df}>0$}\label{chooseAgg}
        \State \Agg \ $\x{a} \leftarrow \x{Agg}[index]$    \label{read-agg}
\color{black} 	\State unsigned int64 $\x{aBefore} \leftarrow \FAA(a.\x{value}, |\x{df}|)$  \Comment{apply operation to \Agg\ $a$}\label{FAAagg}
        \While{$a.\x{last}.\x{after} < \x{aBefore}$ \begingroup \color{myblue} $ \vee \x{aBefore}\geq a.\x{final}$\endgroup} \label{wait}\Comment{wait until my batch has been or can be added to $a$'s list}%
\begingroup \color{myblue}
            \If{$\x{aBefore}\geq a.\x{final}$}\label{if-overflow} \textbf{go to} line \ref{read-agg}\Comment{\Agg\ $a$ overflowed; restart in the current \Agg}%
            \EndIf%
\endgroup
        \EndWhile
        
        \State Batch* $\x{batch} \leftarrow a.\x{last}$\label{read-last}
        \If{$\x{batch}.\x{after} = \x{aBefore}$}\label{delegatedTest}\Comment{if operation the first in its batch, it is the batch's delegate}
		      \State unsigned int64 $\x{aAfter} \leftarrow a.\x{value}$\label{readValue}\Comment{get \Agg\ $a$'s \x{value} at the end of the batch of operations}
            \State unsigned int64 $\x{mainBefore} \leftarrow \FAA(\x{Main}, (\x{aAfter}-\x{aBefore}) \cdot \mbox{sgn}(\x{df}))$ \label{FAAMain} \Comment{apply batch of operations on \x{Main}}%
\begingroup\color{myblue}
            \If{$\x{aAfter} \geq \threshold$} \Comment {this is last batch on \Agg\ $a$}\label{overflow-check}
                \State $\x{Agg}[index] \leftarrow \mbox{new \Agg}(0, \mbox{new Batch}(0,0,0,\perp), \infty)$ \label{newAgg}\Comment{retire $a$ and replace it with a new \Agg}
                \State $a.\x{final} \leftarrow \x{aAfter}$ \Comment{ensure no more batches are performed to $a$}\label{setFinal}%
            \EndIf%
\endgroup
            \State $a.\x{last} \leftarrow \mbox{new Batch}(\x{aBefore}, \x{aAfter}, \x{mainBefore}, \x{batch})$ \label{newBatch} \Comment{create a new Batch and add it to $a$'s list}
            \State \Return{\x{mainBefore}}\label{return1}
        \Else \Comment{this operation is in a Batch already added to $a$'s list}
		      \While{$\x{batch}.\x{before} > \x{aBefore}$}\Comment{find batch with $\x{batch}.\x{before} \leq \x{aBefore} < \x{batch}.\x{after}$}\label{search-begin}
                \State $\x{batch} \leftarrow \x{batch}.\x{previous}$
            \EndWhile\label{search-end}
            \State \Return $\x{batch}.\x{mainBefore} + (\x{aBefore} - \x{batch}.\x{before})\cdot \mbox{sgn}(\x{df})$ \Comment{compute result to return}\label{return2}%
	\EndIf
\EndFunction
\medskip

\Function{int}{Fetch\&AddDirect}{int \x{df}} \Comment{apply \FetchAndAdd\ directly to \x{Main}}
	\State \Return{\FAA(\x{Main}, \x{df})}\label{direct}
\EndFunction

\medskip

\Function{int}{Compare\&Swap}{int \x{old}, int \x{new}}\Comment{any other available primitive can be applied similarly to \x{Main}}
	\State \Return{\op{CAS}(\x{Main}, \x{old}, \x{new})} \Comment{use hardware \op{CAS} directly on \x{Main} }
\EndFunction
\caption{Aggregating Funnels:  a strongly linearizable Fetch\&Add implementation. 
}
\label{code-overflow}
\end{algorithmic}
\end{algorithm*}

\begin{algorithm}
\begin{algorithmic}[1]
\StartFromLine{43}
\Function {int} {Choose\Agg} {int \x{df}}
    \State int $\x{g} \leftarrow \left\lfloor \x{threadIdx} / \sqrt{p} \right\rfloor $
    \If{$\x{df} > 0$}
        \Return $g$
    \Else \xspace
        \Return $m+g$
    \EndIf
\EndFunction
\end{algorithmic}
\caption{One possible implementation of the \op{Choose\Agg} function for $p$ threads using $m=\floor{\sqrt{p}}$ \Aggs\ for each sign.}\label{evenSpread}
\end{algorithm}

\Cref{code-overflow} gives pseudocode for our implementation. 
Shared variable names are capitalized; thread-local variables are not.
We use the notation sgn($x$) for the signum function that returns $1$ if $x>0$, or $-1$ if $x<0$.

We first focus on the
black part of the code. Code in \begingroup\color{myblue}cyan\endgroup\ copes with the (rarer) case of an overflow
on an \Agg\ and it is discussed in \Cref{overflow-discussion}.  We also focus on\linebreak \FetchAndAdd\ operations
with positive arguments first.

A \FetchAndAdd(\x{df}) operation first chooses an \Agg\ $A$ (line~\ref{chooseAgg}) 
from among the $m$ aggregators for the sign of \x{df}.
It then applies a \FAA(\x{df}) primitive
to $A.\x{value}$ (line \ref{FAAagg}).
To reduce contention, several concurrent operations that chose $A$
may be combined into a batch.
The \emph{first} operation of the batch to perform its \FAA\
on $A.\x{value}$ is selected as the batch's \emph{delegate}.
Only the delegate proceeds to access \x{Main}, where it
performs
a \FAA\ that
adds the sum of the batch's arguments to \x{Main}.
To provide a linearization for $O$,
we linearize the entire batch of operations at the batch's \FAA\ 
 on \x{Main}.
Thus, we maintain an invariant that \x{Main} holds 
the value that $O$ would have if all operations linearized so far were performed on it.
Operations within a batch are linearized in the order they performed
their \FAA\ on $A.{value}$.

Only the first operation in each batch at $A$ gets the result of the \FAA\
on $\x{Main}$, so that delegate operation must share this information with the rest of the batch's operations, so that each can figure out the result  it should return.
To facilitate this, \y{$A$ stores a singly-linked list of \emph{Batch objects},}
one for each batch of operations from $A$
that has been applied to \x{Main}. 
$A.\x{last}$ points to
the Batch object 
for the most recent batch of operations applied to $A$.
Each Batch object has a pointer to the previous Batch
and several other pieces of information:  the fields \x{before} and \x{after}
store $A$'s \x{value} before and after the batch of operations is applied to~$A$,
and \x{mainBefore} stores the value of \x{Main} just before the batch
of operations is applied to \x{Main}. Each field of a Batch is immutable.

The most recent Batch $B$ in $A$'s list of Batch objects is used to determine which \FetchAndAdd\ 
operation is the first in $A$'s next batch $B'$.
This delegate operation for $B'$ is the operation whose \FAA\ on $A.\x{value}$ returns the value $A.\x{last.after}$,
i.e., the value that $A$ had after $B$ has been applied to $A$. 
After getting the result \x{aBefore} from its \FAA\ on $A.\x{value}$, an operation \x{op}
checks whether it should wait on line~\ref{wait}. 
The delegate operation $op_{del}'$ for $B'$ will be the only operation among the operations of $B'$
that will eventually evaluate the condition of line~\ref{wait} to true (equality will hold) when $B$ is added to $A$'s list, thus causing $op_{del}'$ to exit the wait loop.
Each of the other operations 
will wait until $A$'s list contains a Batch object 
whose \x{after} field is greater than the value stored in the operation's \x{aBefore} variable.
Such a Batch  will be added to the list by a delegate operation, as we describe below. 

The delegate operation $op_{del}$ of a Batch $B$ executes lines~\ref{readValue}--\ref{return1}. 
Its read of $A.\x{value}$ on line~\ref{readValue} determines the end of Batch~$B$:
$B$~contains all operations that perform their \FAA\ on $A$ from the time $op_{del}$ performed its
\FAA\ on $A$ until $op_{del}$ reads $A.\x{value}$ on line~\ref{readValue}. 
The value read on line~\ref{readValue} will be stored in the 
\x{after} field of the new Batch that $op_{del}$ appends to the list on line~\ref{newBatch}.
This occurs after $op_{del}$ has performed the \FAA\ on $\x{Main}$ at line~\ref{FAAMain},
accomplishing all operations of $B$. 
Since only the delegate thread can add the next Batch to $A$'s list, 
$op_{del}$ can use a simple write at line~\ref{newBatch} to add this Batch.
The \FAA\ on $A$ on line \ref{FAAagg} by the delegate operation $op_{del}'$ of the \emph{next} batch $B'$ after $B$ must be after $op_{del}$ executes line~\ref{readValue}, but may also be before $op_{del}$ adds its new Batch to $A$'s 
list on line~\ref{newBatch}. In the latter case, $op_{del}'$ will wait at line \ref{wait} for $op_{del}$ to add its Batch 
to $A$'s batch list: only then will $op_{del}'$ evaluate the condition of the wait statement on line~\ref{wait} to be false
and proceed to execute lines~\ref{delegatedTest}--\ref{return1} itself. 

A \FetchAndAdd{} operation \x{op} on $O$  computes its result as follows.
If \x{op} is the first operation within its batch on \Agg\ $A$
(i.e., it is the delegate of its batch), it returns the result of its
\FAA\ on \x{Main} at line~\ref{return1};
this is the value that \x{Main} had before the batch's operations are applied to it.
If $op$ is not the delegate of its batch, 
it executes lines~\ref{search-begin}--\ref{return2} to find its response. 
Specifically, it 
first looks for a Batch object
$B$ in $A$'s list with $B.\x{before} \leq \x{aBefore} < B.\x{after}$ (lines \ref{search-begin}--\ref{search-end}).  This is the Batch to which $op$ belongs.
Then, $\x{aBefore}-B.\x{before}$ is the sum of the arguments of operations within the batch that precede \x{op}
(i.e., the operations that performed their \FAA\ on A.value before $\x{op}$).
So, $op$ returns $B.\x{mainBefore}+\x{aBefore}-B.\x{before}$ on line~\ref{return2}. 
Traversing the list is needed because $op$ might be so slow that several batches after $B$ could be added 
to $A$'s list by the time $op$ reads $A.\x{last}$ on line~\ref{read-last}.
In our experiments, we find that $97\%$ of operations locate their batch at the head of the list, thus avoiding looping on lines~\ref{search-begin}--\ref{search-end}.

\remove{An \Agg\ $A$'s list of batches is also used to determine which \FetchAndAdd\ 
operation is the first in its batch.
After getting the result \x{aBefore} from its \FAA\ on $A.\x{value}$, an operation \x{op}
waits at line~\ref{wait} until $A$'s list contains a Batch object 
whose \x{after} field is greater than or equal to \x{aBefore}.
If equality holds, then \x{op} is the first operation in its batch,
and it adds a new Batch object to $A$'s list.
The \x{after} value of this new Batch is determined by reading $A.\x{value}$ on line~\ref{readValue}.
Since only one thread can add the next Batch to $A$'s list, this addition can be accomplished by a simple
write at line~\ref{newBatch}.
}

\begin{figure}
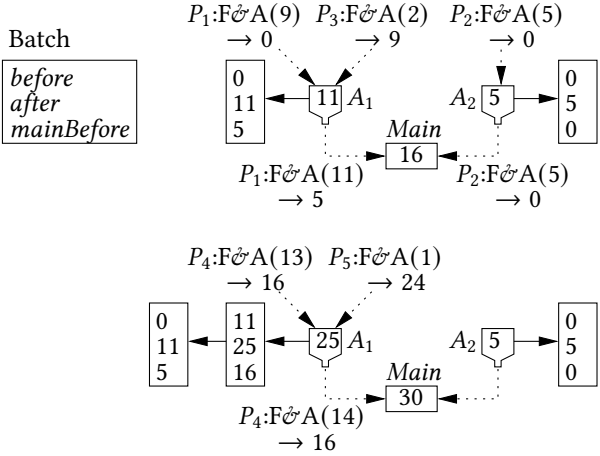
\caption{Example execution with five \FetchAndAdd\ operations and two \Agg\ objects $A_1$ and $A_2$.
\label{counterexample}
\vspace{-0.1in}
}
\end{figure}

\Cref{counterexample} shows an example of how the data
structure evolves when accessed by five \FetchAndAdd\ operations.
The arguments and results of all hardware \FAA\ primitives are shown.
The upper diagram shows the data structure after three operations:  
two as a batch on \Agg\ $A_1$ and a single operation as a batch on \Agg\ $A_2$.  
The lower diagram shows the data structure after another batch of two operations is applied via $A_1$.
The linearization order of the threads' operations is $P_2, P_1, P_3, P_4, P_5$.
The operations by threads $P_1, P_2$, and $P_4$ see that they are the first operations
in their respective batches, since the value they receive from
their \FAA\ on an \Agg\ is that \Agg's \x{value}
after the previous Batch was applied (or 0 if there is no previous Batch).
Therefore, these delegate operations do
a \FAA\ on \x{Main}, while the non-delegate operations by $P_3$ and $P_5$ wait to compute
their results.
$P_3$'s \FAA(2) on $A_1.\x{value}$ returns 9.
It concludes that it belongs to $A_1$'s oldest Batch,
which takes $A_1.\x{value}$ from value 0 to 11.
\x{Main} had the value 5 before that batch was applied to it.
Thus, $P_3$ returns $5+9-0=14$.
The Batch that $P_3$ needs to compute its result remains accessible in $A_1$'s
 list of Batches, even after other
Batches are added, so $P_3$ can compute its result even if
it is delayed before accessing this list.
Similarly, $P_5$ finds the
Batch in $A_1$'s list containing 24 and computes its result as
$16+24-11=29$.

To handle \FetchAndAdd\ operations with negative arguments, the \Aggs\ are partitioned into
$m$ positive  and $m$ negative \Aggs.  A \FetchAndAdd\ 
chooses an \Agg\ of the type  matching its argument's sign.  (A \FetchAndAdd(0) simply
reads the \x{Main} variable; see line~\ref{FAA0}.)
\e{To simplify the code, when a \FetchAndAdd(\x{df})  applies its \FAA\ on an \Agg's $value$ field,
it uses the absolute value of \x{df}.  This ensures that the \Agg's value only increases, making it easy 
to determine which Batch an operation belongs to.
When a batch of operations from an \Agg\ for negative operands is applied to \x{Main},
we multiply the operand of the \FAA\ on \x{Main} by $-1$ at line \ref{FAAMain}.
Similarly, the sign has to be taken into account when computing the result of a non-delegate
operation with a negative operand at line \ref{return2}.
One effect of splitting operations according to their sign is that, even if the value
of the implemented object $O$ remains small, the $value$ fields of \Aggs\ may grow without bound.
We describe in \Cref{overflow-discussion} how the cyan code handles this.}

We prove in Section~\ref{proof} that \Cref{code-overflow} is linearizable, regardless of the number of \Aggs\ and
how  \Aggs\ are chosen at line~\ref{chooseAgg}.
Thus, these choices can be tuned to achieve good performance.
\Cref{evenSpread} shows one straightforward way to do this.
It divides the $p$ threads that access
the \op{Fetch\&Add} object $O$ into $\sqrt{p}$ groups of $\sqrt{p}$ threads each,
and assigns each group to one of the $m$ \Aggs\ of each type.
This limits contention on any shared variable to $\sqrt{p}$,
because at most $\sqrt{p}$ threads access each \Agg, and
at most one thread from each of the $\sqrt{p}$ groups can access \x{Main}
at any one time.
Operations could also be assigned to a random \Agg\ of the appropriate type.
We discuss how these choices were made for our experiments in \Cref{allocation-scheme}.

Any other operations that can be applied atomically to a memory word can also
be applied to our object $O$, simply by applying them directly to \x{Main}.
In \Cref{code-overflow}, we show the code for performing a \op{Read} or \op{Compare\&Swap}, but other
operations  would work in exactly the same way.
Similarly, we also provide a \op{Fetch\&AddDirect} that applies a \FAA\ directly to \x{Main},
which can be used by high-priority threads to perform their operations
with lower latency.

\subsubsection{Handling Overflows in \Aggs}
\label{overflow-discussion}

The code shown in \begingroup \color{myblue}cyan\endgroup\ in Algorithm~\ref{code-overflow}
copes with the case that an \Agg's \x{value} field may overflow. 
When a Batch of \FetchAndAdd\ operations increases the \x{value} field of an \Agg\ $A$
beyond \threshold, defined on line \ref{rep-max} to be $2^{63}$,
$A$ is retired and replaced by a new \Agg.
We shall show that each of the $p$ threads can do at most one \FetchAndAdd\ on $A.\x{value}$ after crossing the threshold.
So, provided every argument to \FetchAndAdd\ is strictly less 
than $2^{63}/p$ in absolute value, this ensures that the value of $A.value$
never reaches $2^{64}$ and causes an overflow error.
\e{More generally, if we have a bound of $B$ on the 
arguments of \FetchAndAdd\ operations, we could instead define \threshold\ to be $2^{64}-p\cdot B$.}

This protects against overflow  in individual \Aggs.
If, however, the value of the implemented object $O$ overflows, then \x{Main} will overflow too,
and $O$ will behave in the same way that an overflow of a hardware fetch-and-add object would behave.
For example, if the hardware \FAA\ instructions wrap around when an overflow occurs, 
then \x{Main}'s value will wrap around, and so will $O$'s.  
\e{(We assume that the arithmetic in line \ref{return2} wraps around modulo $2^{64}$.)}

We now describe how an \Agg\ object $A$ gets retired after $A.\x{value}$ surpasses \threshold.
Delegate operations check on line \ref{overflow-check} whether the $\x{value}$ field of $\x{A}$ exceeds \threshold.
If so, the delegate retires $A$, meaning that $A$ cannot be used for any more batches of operations after $B$.
It does this by creating a new
\Agg\  on line \ref{newAgg} to replace $A$ in the $Agg$ array and then
setting $A.\x{final}$ on line \ref{setFinal} to $A$'s \x{value} after $B$.

Consider an operation $op$ that is too late to join $A$'s final Batch $B$,
i.e., $op$ performs its \FAA\ on $A.value$ on line~\ref{FAAagg} after the delegate of $B$ reads $A.value$ on line~\ref{readValue}.
We argue that $op$ will eventually go back to line \ref{read-agg} and use a different Aggregator.
$B$'s delegate sets $A.\x{final}$ on line \ref{setFinal} before it appends $B$ to $A$'s list of Batches on line \ref{newBatch}.  
Thus, until $B$'s delegate executes line \ref{setFinal}, 
$A.\x{last}.\x{after}$ is less than $op$'s value of \x{aBefore}, and $A.\x{final}=\infty$, so $op$ remains in its loop at line \ref{wait}. 
After $B$'s delegate sets $A.\x{final}$, the test on line \ref{if-overflow} ensures that $op$ goes back
to line~\ref{read-agg}, and will not access $A$ again since $B$'s delegate replaced
$A$ with a new \Agg\ at line \ref{newAgg}.

We now argue that each thread $q$ does at most one \FAA\ on $A.\x{value}$ that
sets $A.\x{value}$ to a value greater than \threshold, which ensures that $A.\x{value}$ never overflows.
Suppose a thread does a \FAA\ at line \ref{FAAagg} of some operation $op$ that changes $A.\x{value}$ to a value greater than \threshold.
As argued above, $op$ cannot do another \FAA\ on $A.\x{value}$; if
the operation returns to line \ref{FAAagg} again, it accesses a different \Agg.
Either \x{op} belongs to the final Batch $B$ of $A$, or \x{op} performs its \FAA\ on $A.\x{value}$
too late to join that final Batch.
In either case, $B$'s delegate must retire $A$ before \x{op} can complete, and so the \emph{next} \FetchAndAdd\ 
by the thread $q$ will not use $A$.

The test on line \ref{wait} is a bit subtle when the second disjunct is added to handle overflows.
It requires reading two locations in shared memory:
$a.\x{last}$ and $a.\x{final}$.  
If a \FetchAndAdd\ operation $op$ exits the while loop and reaches line~\ref{read-last}, then
$a.\x{last}.\x{after} \geq \x{aBefore}$ at the first of the two reads (since the \x{after} field of the Batch $a.\x{last}$ is immutable)
and $\x{aBefore}<a.\x{final}$ at the second of the two reads.
If $a.\x{last}.\x{after}$ is strictly greater than \x{aBefore}, then $op$ belongs to a Batch $B$ whose delegate
has added $B$ to $A$'s list, so $op$ will continue on to determine its result using lines \ref{search-begin}--\ref{return2}, as in the case without overflow handling.
If $a.\x{last}.\x{after}$ is equal to $\x{aBefore}$, $op$ is the delegate of its batch of operations,
and the preceding batch did not retire $A$ (because $op$'s test saw that $\x{aBefore}<a.\x{final}$), 
so it is safe for $op$ to add a new Batch to $A$'s batch list
using lines \ref{readValue}--\ref{return1}, as in the case without overflow handling.

\subsubsection{Memory Management and Space Usage}
Our implementation in~\Cref{sec:experiments} uses epoch-based reclamation \cite{F03} for Batch and Aggregator objects.
Other safe memory reclamation techniques would also work.
We cannot prove any worst-case bound on memory usage when using epoch-based reclamation, however we can bound the number of objects that have been allocated and not yet retired to the epoch-based collector.
An Aggregator is retired as soon as it is no longer pointed to by the $\x{Agg}$ array and a Batch is retired as soon as it is not pointed to by an Aggregator.
Therefore, there are at most $\Theta(m)$ Aggregator and Batch objects that have not yet been retired.
In addition to these, we also use $\Theta(m)$ memory words to store the $\x{Main}$ and $\x{Agg}$ variables. So the overall space usage, if we do not count objects that have been retired and not freed, is $\Theta(m)$.

As a sidenote, for a counter, which supports only \op{Add} and \op{Read} operations, we can save space by not using Batch objects at all---if each \Agg\ simply stores
the value that would usually be stored in \x{last.after}, \op{Add} operations can detect when to stop waiting
for their batch to be applied to \x{Main} (as in line~\ref{wait} of \op{Fetch\&Add}).
This simplicity stems from the fact that an \op{Add} need not 
figure out a response value. 

\subsection{Applying the Construction Recursively}
\label{recursive-sec}

As described above, using $m=\sqrt{p}$   
reduces contention on any variable in a $p$-thread system
to $O(\sqrt{p})$.  If $p$ is very large, one can reduce contention even further by applying the construction
recursively.  We can replace \x{Main}
or any of the \Aggs' \x{value} fields by an instance of \Cref{code-overflow}.
We can repeat this process  
to any desired depth of recursion.

For example, consider a fetch-and-add object $O$ for $p$ threads implemented using \Cref{code-overflow} where we replace 
\x{Main} in $O$ by another instance $O'$ of \Cref{code-overflow}.
We use $m=p^{2/3}$ for $O$ and $m'=p^{1/3}$ for $O'$.
Suppose threads choose \Aggs\ as shown in \Cref{recursive}.  (For simplicity, the figure shows only the \Aggs\ for positive arguments.)
Contention on each \Agg\ of $O$ is at most $p/m = p^{1/3}$.
Contention on each \Agg\ of $O'$ is at most $m/m' = p^{1/3}$.
Contention on the variable $\x{Main}'$ of $O'$ is at most $m'=p^{1/3}$.
Thus, we have reduced contention on all variables to $O(p^{1/3})$.

Repeating this process of replacing \x{Main} by another instance of \Cref{code-overflow} $k$ times reduces
contention on any base object to $O({p}^{1/(k+1)})$.
Taking $k=\log_2 p$ reduces
contention to $O(1)$ using $O(p)$
\Aggs\ in total.
A \FetchAndAdd\ operation would access at most $O(\log p)$ base objects.

Alternatively, we can repeatedly replace \emph{both} \x{Main} and all \Aggs' \x{value} fields
by \Cref{code-overflow}.
Doing this $k$ times reduces contention on any location to $O(p^{1/2^k})$.
Taking $k=\log\log p$ yields $O(1)$ contention using $O(p)$ Aggregators.

There is a tradeoff:  reducing contention on individual locations
requires a \FetchAndAdd\ to access more locations (or wait for others to do so).
Moreover, when a \FetchAndAdd\ operation must access more locations,
it spends a smaller fraction of its time at each one, so it is 
less likely to contribute to contention at that location at any particular time.  
Thus, the actual contention at a location will typically be smaller than the worst-case upper bound.  
So, it is impractical to try to reduce the worst-case contention too much:
this will cost time (to access more \Aggs) without the payoff of reducing contention in practice.
Indeed, our experiments revealed no advantage of using even a single replacement
(as shown in \Cref{recursive}) for values of $p$ up to 176.  The recursive construction would pay off only for very large thread counts.

\begin{figure}
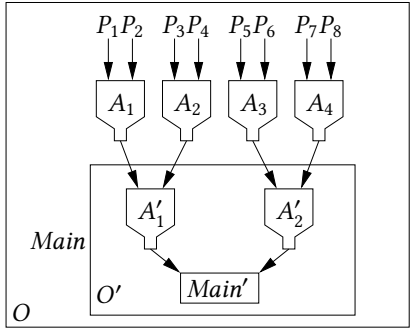
\caption{Example of recursive construction with $p=8$.\label{recursive}
\vspace{-0.2in}
}    
\end{figure}


\subsection{Correctness}
\label{proof}

We prove that \Cref{code-overflow} is linearizable.
Each operation is linearized when it is applied to \x{Main}, either 
as part of a batch in the case of \op{Fetch\&Add}, or
individually in the case of the other operations
(\op{Read}, \op{Fetch\&AddDirect}, and \op{Compare\&Swap}).
We show in \Cref{main-inv} that this ensures that \x{Main} always stores the true value that the implemented fetch-and-add object should have. 
We must show that the effect of each \FetchAndAdd\ operation $\x{op}$ is applied to \x{Main}
exactly once and that $\x{op}$'s response is consistent with the linearization (\Cref{consistent}). 
We must also show (in \Cref{real-time}) that \x{op}'s linearization point is between its invocation and response.
Since our linearization points can be identified as the execution unfolds, without knowledge of later events,
the implementation is also strongly linearizable~\cite{GHW11}.

\Siddhartha{I like the proof that you drew up here, Eric. My only additional though: We could add references to the Lemmas below for each of the steps that are mentioned in the intro paragraph above, so that the less theoretically oriented PPoPP reviewers can skip the rest of the section with confidence that things were proved methodically.}
\Y{I agree with Siddhartha. Also, the summary above is probably too high level. It would be nice if we come up
with a summary that refers to the different lemmas and their role in the entire proof.}
\Eric{I added refs to lemmas, but time (and space) is running short to write more of a high-level proof; I think we did try to explain the intution about why the algorithm works in the algorithm description, so I hope that's enough.}

We first prove the following invariant, which  ensures that $A$'s Batch list is sorted
by \x{before} fields and that the \x{before} field of one Batch matches the previous Batch's \x{after} field.

\begin{invariant}
\label{list-inv}
Let $A$ be an \Agg\ object.
If the list of Batch objects reachable from $A.\x{last}$ by following \x{previous} pointers
is $B_k, B_{k-1}, \ldots, B_0$, then
$A.\x{value} \geq B_k.\x{after}$, for $1\leq j\leq k$, $B_j.\x{after} > B_j.\x{before} = B_{j-1}.\x{after}$, 
and $B_0.\x{after} =0$.
\end{invariant}
\begin{proof}
Initially, $k=0$ and $A.\x{value} = 0 = B_0.\x{after}$.

Batch $B_k$ is written to $A.\x{last}$ at line \ref{newBatch} of some \FetchAndAdd.
$B_k.\x{after}$ is read from $A.\x{value}$ on line \ref{readValue}.
Since all \FetchAndAdd\ operations applied to $A$ have positive arguments and $A.\x{value}$ never overflows by the argument in Section \ref{overflow-discussion}, 
$A.\x{value}$ can only increase. 
Thus, $A.\x{value} \geq B_k.\x{after}$.

Consider any Batch $B_j$ created at line~\ref{newBatch}.  $B_j.\x{before}$ is the value 
\x{aBefore}
that the \FAA\ on $A.\x{value}$ at line~\ref{FAAagg} returned,
and $B_j.\x{after}$ is the value read from $A.\x{value}$ at line~\ref{readValue}.
$A.\x{value}$ only increases, so $A.\x{value}$ at line~\ref{readValue} is strictly larger
than the result of the \FAA\ at line~\ref{FAAagg} (since the \FAA's argument is not 0, by the test on line~\ref{FAA0}).
Hence, $B_j.\x{after} > B_j.\x{before}$.
Moreover, line~\ref{newBatch} sets $B_j.\x{previous}$ to \x{last}, and by the test on line~\ref{delegatedTest},
$\x{last}.\x{after}$ is equal to the value \x{aBefore} stored in $B_j.\x{before}$.
Thus, $B_j.\x{before} = B_{j-1}.\x{after}$.
\end{proof}

\Eric{I tried to rewrite the following paragraph to address Youla's concern that some of it was hard to understand.}

We now define linearization points more formally.
\Read, \FetchAndAddDirect, \CompareAndSwap, and \FetchAndAdd(0) \linebreak are each linearized when they access \x{Main}.
We linearize the \FetchAndAdd\ operations with non-zero arguments as follows.
Whenever a delegate \FetchAndAdd\ operation \x{op} that chose an \Agg\ $A$
performs a \FAA\ on \x{Main} at line \ref{FAAMain},
we linearize all operations in \x{op}'s batch (in the order of their \FAA\ operations on $A$.\x{value}).
Recall that the operations in \x{op}'s batch are those that perform a
\FAA\ on $A.\x{value}$ during the interval of time between \x{op}'s accesses to $A.\x{value}$ on
lines \ref{FAAagg} and \ref{readValue} (including \x{op} itself).
The \x{before} and \x{after} fields of the Batch that \x{op} adds to $A$'s list store the
values $A.\x{value}$ had at the beginning and end of this interval.
It follows from \Cref{list-inv} that the intervals for two delegate operations that used the \Agg\ $A$ do not overlap.
Therefore, each operation is assigned a unique linearization point.

\begin{lemma}\label{real-time}
Each operation is linearized between its invocation and its response.
\end{lemma}
\begin{proof}
The claim is trivial if the operation is linearized at its own step.
So, for the remainder of the proof, consider a non-delegate \FetchAndAdd\ $\x{op}'$ that is linearized at the \FAA\ on \x{Main} by its Batch's delegate operation \x{op}.
Let $A$ be the \Agg\ chosen by \x{op} and $\x{op}'$.
By definition, $\x{op}'$ performed a \FAA\ on $A.\x{value}$ before $\x{op}$'s \FAA\ on \x{Main},
so the linearization point of $\x{op}'$ is after $\x{op}'$ is invoked.
Since $\x{op}'$ is not a delegate operation, it cannot terminate on line \ref{return1}.
So, suppose $\x{op}'$ terminates at line~\ref{return2}.
Since $\x{op}'$ completed the waiting loop at line \ref{wait},
some operation added a Batch $B$ to $A$'s list with $B.\x{after}$ strictly greater  than
the result of the \FAA\ $\x{op}'$ performed on $A.\x{value}$.
It follows from \Cref{list-inv} that \x{op} is the operation that added the \emph{first} such
Batch to $A$'s list, which must have happened before $\x{op}'$ completed its waiting loop.
\Y{Justify in a more detailed way the claim in the which part of the previous sentence.
BTW, the claim is correct but it would be nice if the reader does not have to think in order to verify it.}
Thus,  $\x{op}'$ is linearized when $\x{op}$ performs its \FAA\ on \x{Main}, which is
before $\x{op}'$ terminates.
\end{proof}

We prove the following key invariant by induction.

\begin{invariant}\label{main-inv}
At all times $t$,
\x{Main} stores the value that $O$ would have if all operations linearized before $t$ were performed sequentially in the
order of their linearization points. 
\end{invariant}

\begin{proof}
\noindent
\textit{Base case.} The invariant holds initially, since $\x{Main}=0$.

\noindent
\textit{Inductive step.}
We show that the invariant is preserved by each step that accesses \x{Main}. (Only these steps  are linearization points.)
This is clear for accesses to \x{Main} by all operations other than \FetchAndAdd\ operations with non-zero arguments.
So, consider a \FetchAndAdd\ operation \x{op} that chooses an \Agg\ $A$ for positive arguments and performs a $\FAA(\x{aAfter}-\x{aBefore})$ on \x{Main} at line \ref{FAAMain} using the value \x{aAfter} obtained by reading $A.\x{value}$ at
line \ref{readValue} and the value \x{aBefore} obtained from its \FAA\ on $A.\x{value}$ at line
\ref{FAAagg}.
The \FetchAndAdd\ operations linearized at \x{op}'s \FAA\ on \x{Main} are exactly those that 
perform their \FAA\ on $A.\x{value}$ in between these two steps, so the sum of their
arguments is exactly $\x{aAfter}-\x{aBefore}$.
Thus, this \FAA\ on \x{Main} preserves the invariant.
The argument is similar if the \Agg\ is for negative arguments:  in this case, \x{op}
performs $\FAA(-(\x{aAfter}-\x{aBefore}))$ on \x{Main} at line \ref{FAAMain}.
\end{proof}

\begin{lemma}\label{consistent}
Each operation's response is consistent with the linearization.
\end{lemma}
\begin{proof}
Operations other than \FetchAndAdd\ are
linearized at their access to \x{Main}, so the claim is immediate from \Cref{main-inv}.
Consider a \FAA\ on \x{Main} that is the linearization point of a batch of \FetchAndAdd\ operations $\x{op}_1, \ldots, \x{op}_k$ with arguments $\x{df}_1, \ldots, \x{df}_k$ that all chose the same \Agg\ $A$ (in the order they perform their \FAA\ on $A.\x{value}$).
Then, $\x{op}_1$ is the operation that performs the \FAA\ on \x{Main} with argument $\sum_{i=1}^k \x{df}_i$.
Let $B$ be the Batch object that $\x{op}_1$ creates on line \ref{newBatch}.
Then $B.\x{mainBefore}$ is the value returned by $\x{op}_1$'s \FAA\ on \x{Main}
and $B.\x{before}$ is the value returned by $\x{op}_1$'s \FAA\ on $A.\x{value}$.
Then, $\x{op}_j$ gets the result $\x{bef}_j = B.\x{before} + \sum_{i=1}^{j-1}\x{df}_i$ from its \FAA\ on $A.\x{value}$.
By \Cref{main-inv}, $\x{op}_j$'s response should be $B.\x{mainBefore} + \sum_{i=1}^{j-1}\x{df}_i = B.\x{mainBefore} + \x{bef}_j - B.\x{before}$, which is the value $\x{op}_j$ returns on line \ref{return2}.
\end{proof}

\Cref{real-time,consistent} establish the following main result.

\begin{theorem}
Algorithm~\ref{code-overflow} is a strongly linearizable implementation of a \FetchAndAdd\ object.
\end{theorem}

Since we can always replace an atomic object by a linearizable implementation, it follows that the recursive constructions described in \Cref{recursive-sec} are also linearizable.



\section{Experimental Evaluation}
\label{sec:experiments}

The goals of our experiments are to explore different parameter choices for Aggregating Funnels (\Cref{allocation-scheme}), compare Aggregating Funnels with hardware \FAA\ and the fastest existing  software \FetchAndAdd{} (\Cref{sec:main_benchmark}), explore the effectiveness of using \op{Fetch\&AddDirect} to speed up high-priority threads (\Cref{sec:high-priority}), and observe the performance when we deploy Aggregating Funnels in a state-of-the-art concurrent queue (\Cref{sec:queue-benchmark}).

\subsection{Experimental Setup}
\label{sec:experiment-setup}

\begin{figure*}
    \begin{minipage}{\textwidth}
    \includegraphics[width=0.8\textwidth]{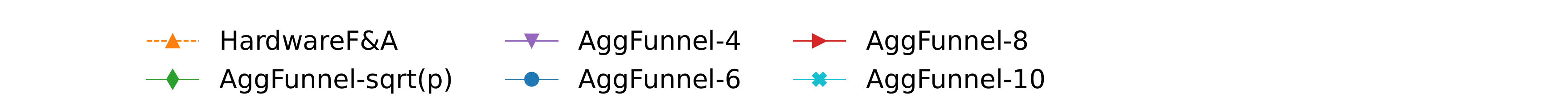}
    \end{minipage}
    \\
    \begin{minipage}{0.32\textwidth}
    \includegraphics[width=\linewidth]{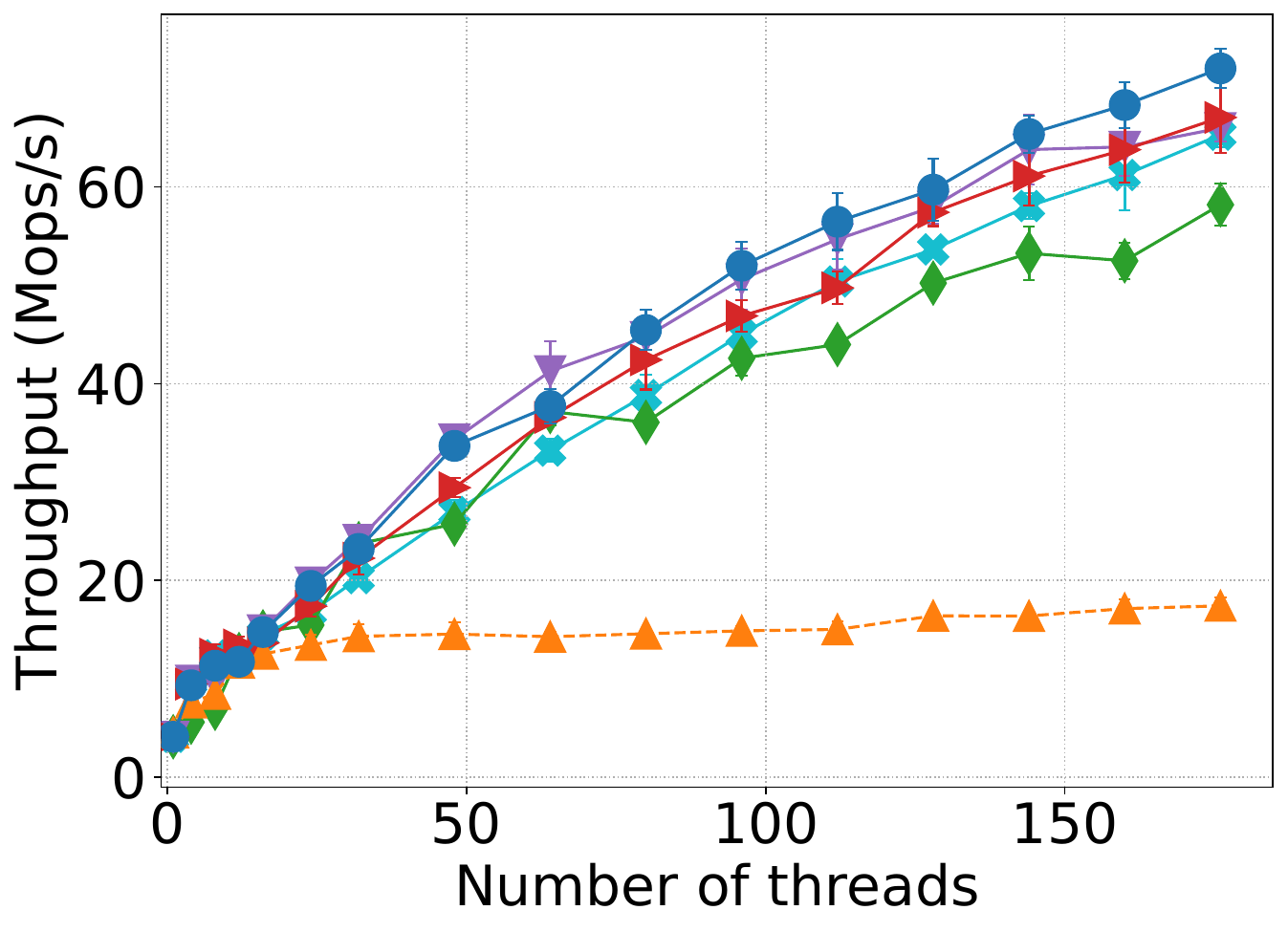}
    \subcaption{90\% \FetchAndAdd{}, throughput}
    \label{fig:child_thr_90}
    \end{minipage}
    \begin{minipage}{0.32\textwidth}    \includegraphics[width=\linewidth]{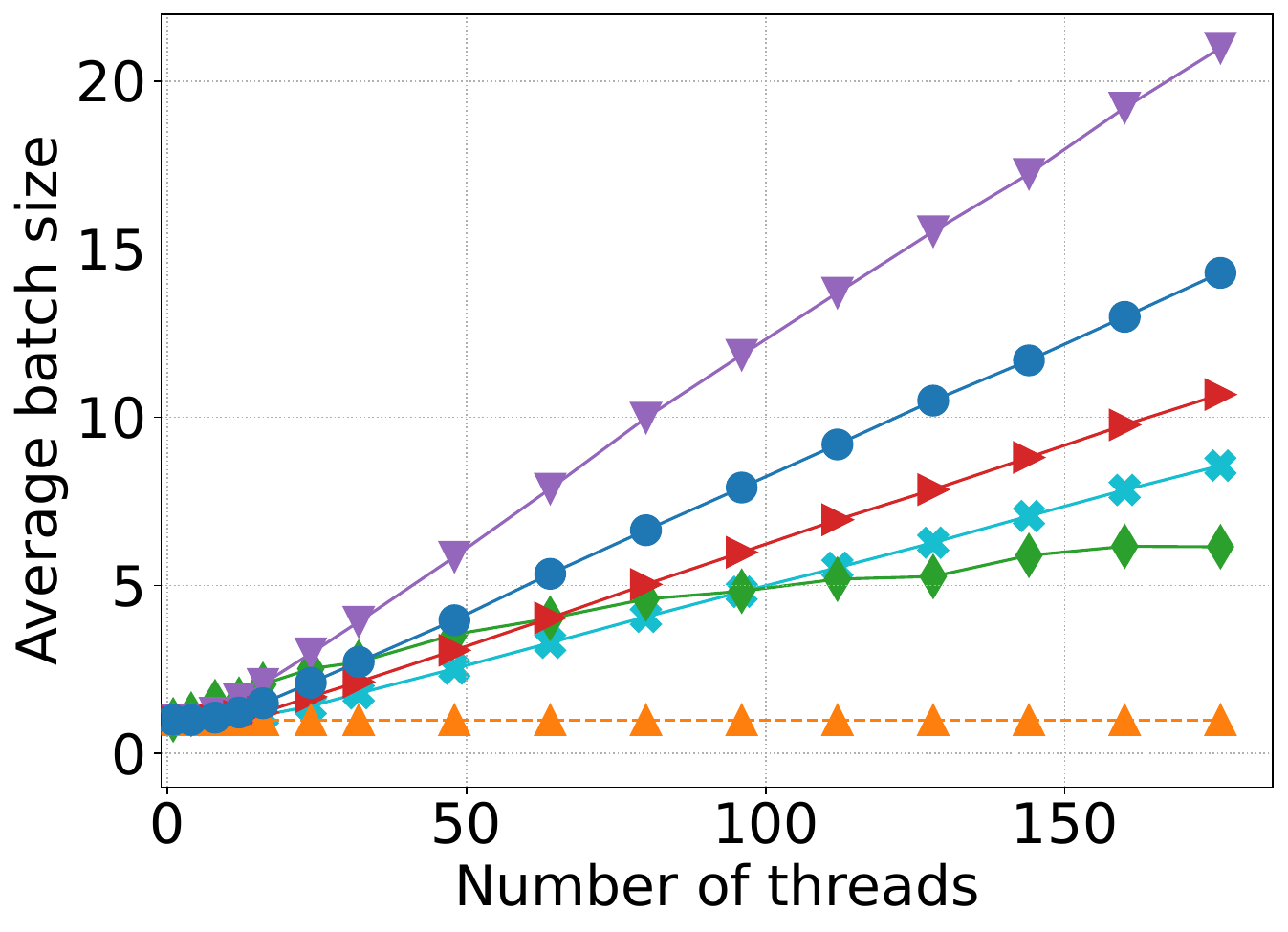}
    \subcaption{90\% \FetchAndAdd{}, batch size}
    \label{fig:child_batch_90}
    \end{minipage}
    \begin{minipage}{0.32\textwidth}    
    \includegraphics[width=\linewidth]{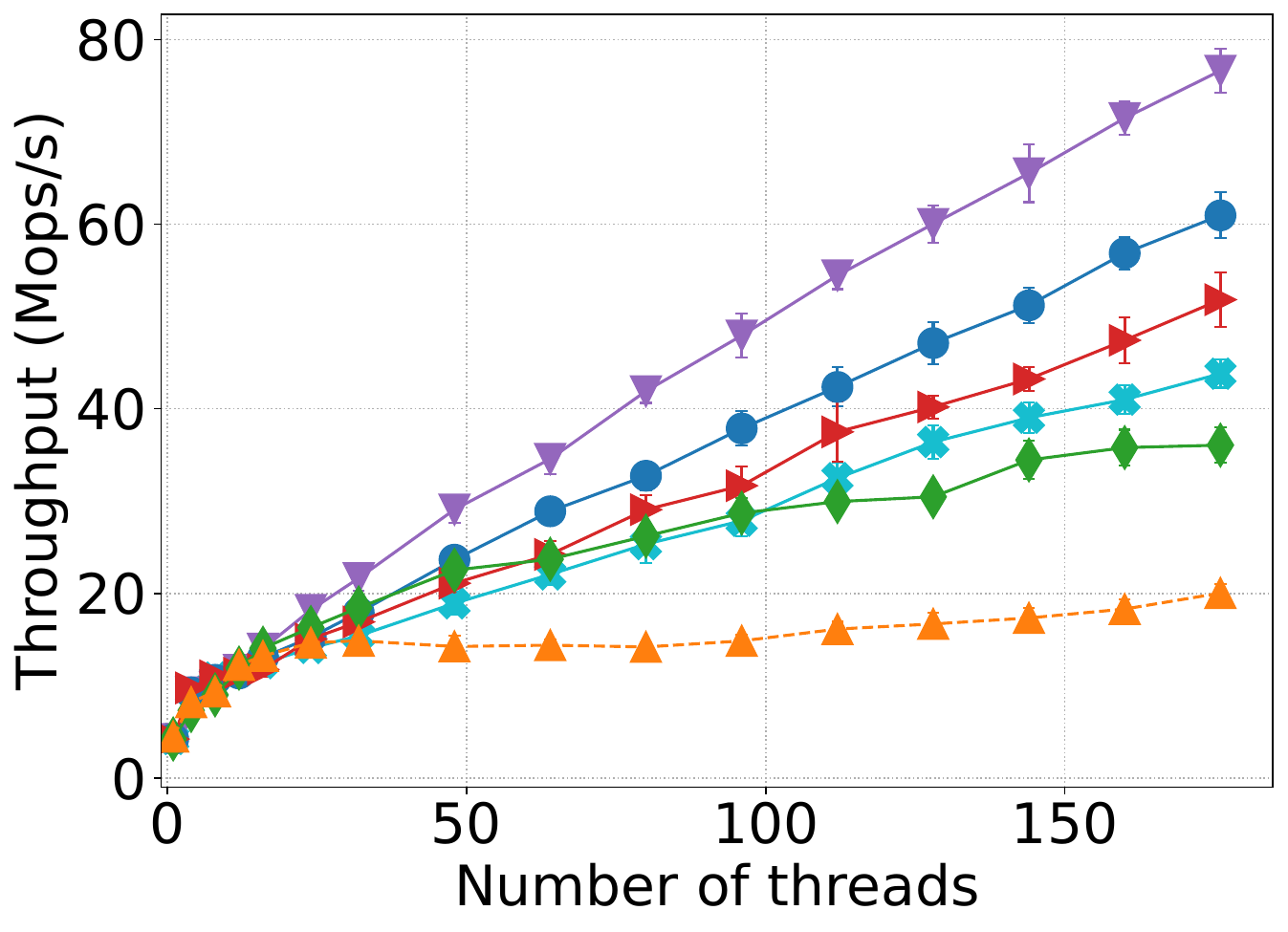}
    \subcaption{50\% \FetchAndAdd{}, throughput}
    \label{fig:child_thr_50}
    \end{minipage}
    \caption{\FetchAndAdd\ performance with different numbers of \Aggs{}.
    \vspace{-0.1in}
    }
    \label{fig:child_main}
\end{figure*}

We used Google Cloud Platform's \texttt{c3-standard-176} instance, which has four 4th Gen
Intel Xeon Platinum 8481C processors with a total of 176 hyper-threads with 2-way hyper-threading, and 704GB of main memory.
We also briefly discuss results on an AMD machine and older Intel machines at the end of \Cref{sec:main_benchmark}.
Our \FetchAndAdd\ and queue benchmarks are implemented in C++, and compiled with \texttt{g++ 13.2.0} with the \texttt{-O3} and \texttt{-std=c++17} flags. We used \texttt{mimalloc} for scalable memory allocation and \texttt{numactl -i all} to distribute memory evenly across the four sockets.

We ran experiments with the simpler version of \Cref{code-overflow} without the code in \begingroup\color{myblue}cyan\endgroup\ for handling overflows.
We believe the overhead added by the overflow handling code should be insignificant in the common case where overflows are infrequent.
We used the appropriate memory fences for correctness in weak memory models,
and memory alignment 
to avoid false sharing.
Our implementation uses epoch-based reclamation \cite{F03} to safely free shared memory.


All \FetchAndAdd\ benchmarks were run for 2 seconds with random arguments between 1 and 100, and with 10 repetitions to average the results.
The error bars in each plot show the standard deviation of the 10 runs, which was small in most cases.
To model a context where a fetch-and-add object is used in a larger algorithm,
we added a geometrically distributed random amount of additional local work between a thread's operations on the 
object.
We varied 
the ratio between  \op{Read}() and \op{Fetch\&Add} operations, the number of threads, and the amount of additional work.
Unless stated otherwise, experiments 
used
a mean of 512 hardware cycles, or roughly 0.2 microseconds, of additional work between operations on the fetch-and-add object.

We measured the \emph{throughput}, i.e., the total number of operations across all of the threads per unit time, of each algorithm to compare their performance. 
We also collected several auxiliary measurements to further understand their behavior, from which we derived two significant metrics.
\emph{Average batch size} is the average number of operations that are aggregated into one \FAA\ on \x{Main}.
Larger batch sizes imply less contention on \x{Main}.
As our \emph{fairness} metric, we use the ratio between the minimum and maximum number of operations completed by a thread. 
Lower fairness indicates that different threads have highly imbalanced throughput.


\subsection{Choosing Number of Aggregators}
\label{allocation-scheme}

The number of \Aggs\ can change the behavior of Aggregating Funnels in various workloads.
Having more \Aggs\ will increase contention on \x{Main}, but it will reduce contention at each \Agg{}'s \x{value}.
The optimal balancing point may vary depending on ratio of \op{Read} and \FetchAndAdd\ operations in workload since \op{Read} operations also contend on \x{Main},
and on the number of threads where hardware \FAA\ reaches its maximum throughput.

In our graphs, \op{AggFunnel}-$m$ denotes the Aggregating Funnels with $m$ \Aggs\ for positive arguments.
(We did not use the $m$ \Aggs\ \e{for negative arguments since all arguments in our experiments} were positive.)
In this section, we study how varying $m$ affects performance.
We use a simple scheme for assigning operations to \Aggs\ that is static and symmetric, which means that a thread chooses the same \Agg\ for all of its operation, and threads are distributed evenly so that 
the maximum contention at different \Aggs\ differ by at most one.

To balance the maximum contention at all \Aggs{} and \x{Main}, we tried $m=\sqrt{p}$ (where $p$ is a known upper bound on the number of active threads), which yields $\sqrt{p}$ maximum contention at all locations.
We also tested with constant values of $m$ for all thread counts $p$, which ensures the maximum contention on the \x{Main} variable is bounded by the constant $m$, while \Aggs\ have maximum contention $p/m$.

\Cref{fig:child_thr_90} and \Cref{fig:child_thr_50}
compare the results for workloads with 90\% and 50\% \FetchAndAdd{}, respectively.
Regardless of the number of \Aggs{}, our algorithm outperforms the hardware \FAA\ from around 20 threads, and the best performing models (i.e. $m=4,\ 6$) are more than 3 times faster than the hardware \FAA\ at 176 threads.

\Cref{fig:child_batch_90} shows that schemes with fewer \Aggs\ have larger batches.
This matches our intuition, since schemes with fewer \Aggs\ have more threads contending on each, and so more threads apply \FAA\ to the \Agg's \x{value} before the delegate thread creates a batch.
While having larger batches means more operations are applied with a single \FAA\ instruction on \x{Main},
having more threads in each \Agg\ slows down the delegate's read (line \ref{readValue}), which proportionally reduces the rate of batch creation.


In contrast to the similar throughput of different schemes in \Cref{fig:child_thr_90} for the 90\% \FetchAndAdd{} workload, \Cref{fig:child_thr_50} shows varying $m$ produced different throughputs for the 50\% \FetchAndAdd{} workload.
All \op{Reads}  access \x{Main}, so read-heavy workloads perform better when \x{Main} has less contention, and schemes with fewer \Aggs{} perform better.

We chose $m=6$ as the default for the rest of the experiments we present, as it outperforms other choices in update-heavy and queue benchmarks in later sections, while performing sufficiently well in other workloads.

\subsection{Fetch-and-Add Benchmark}
\label{sec:main_benchmark}

\begin{figure*}
    \begin{minipage}{\textwidth}
    \includegraphics[width=0.8\textwidth]{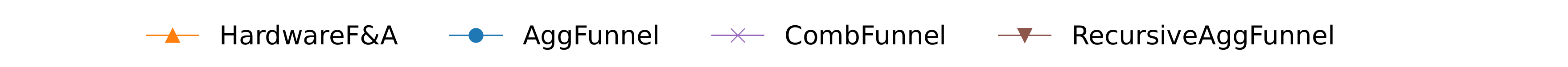}
    \end{minipage}
    \\
    \begin{minipage}{0.32\textwidth}
    \includegraphics[width=\textwidth]{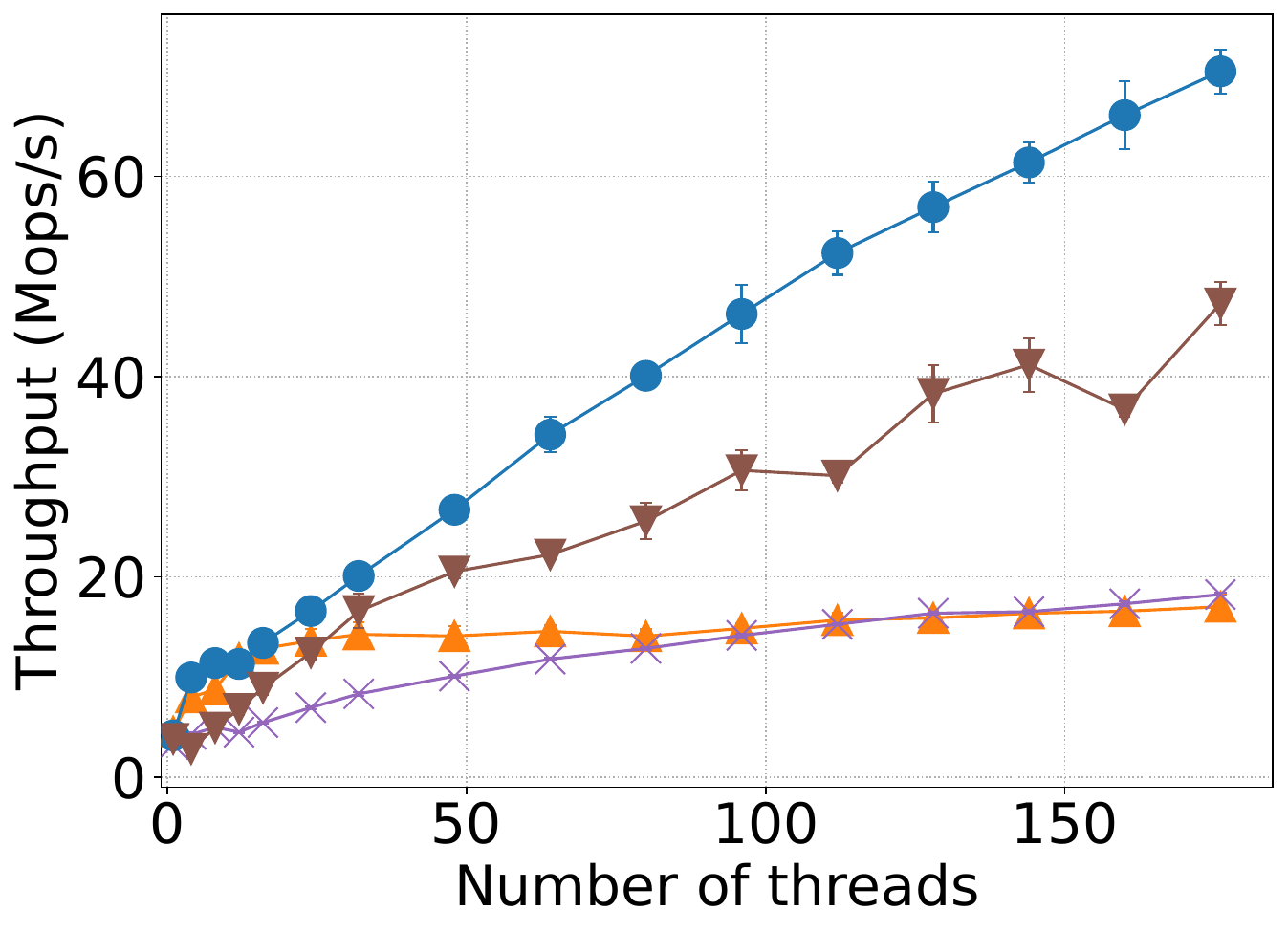}
    \subcaption{90\% \FetchAndAdd{}, 512 cycles, throughput}
    \label{fig:count_thr}
    \end{minipage}
    \begin{minipage}{0.32\textwidth}
    \includegraphics[width=\textwidth]{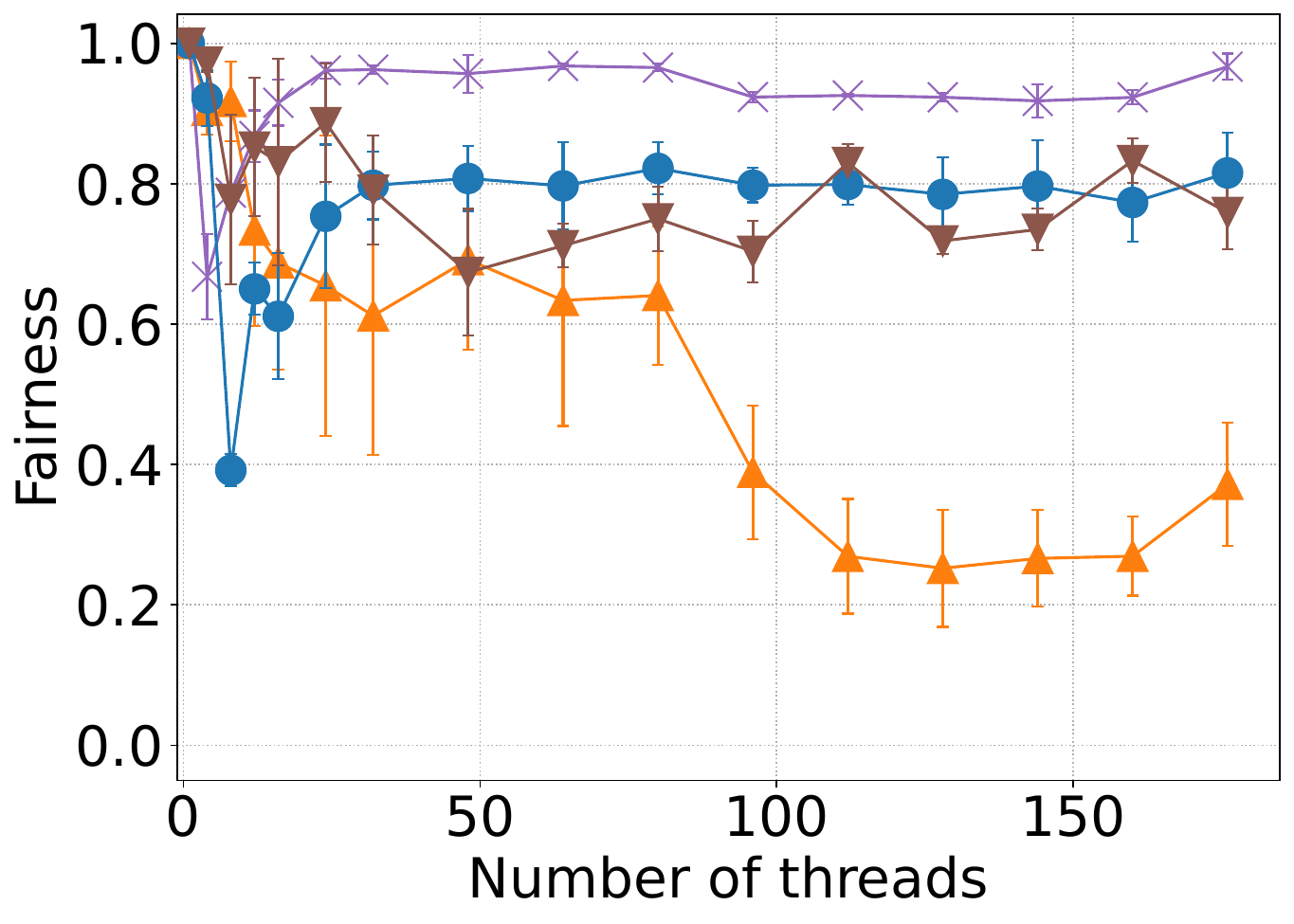}
    \subcaption{90\% \FetchAndAdd{}, 512 cycles, fairness}
    \label{fig:count_fair}
    \end{minipage}
    \begin{minipage}{0.32\textwidth}    
    \includegraphics[width=\textwidth]{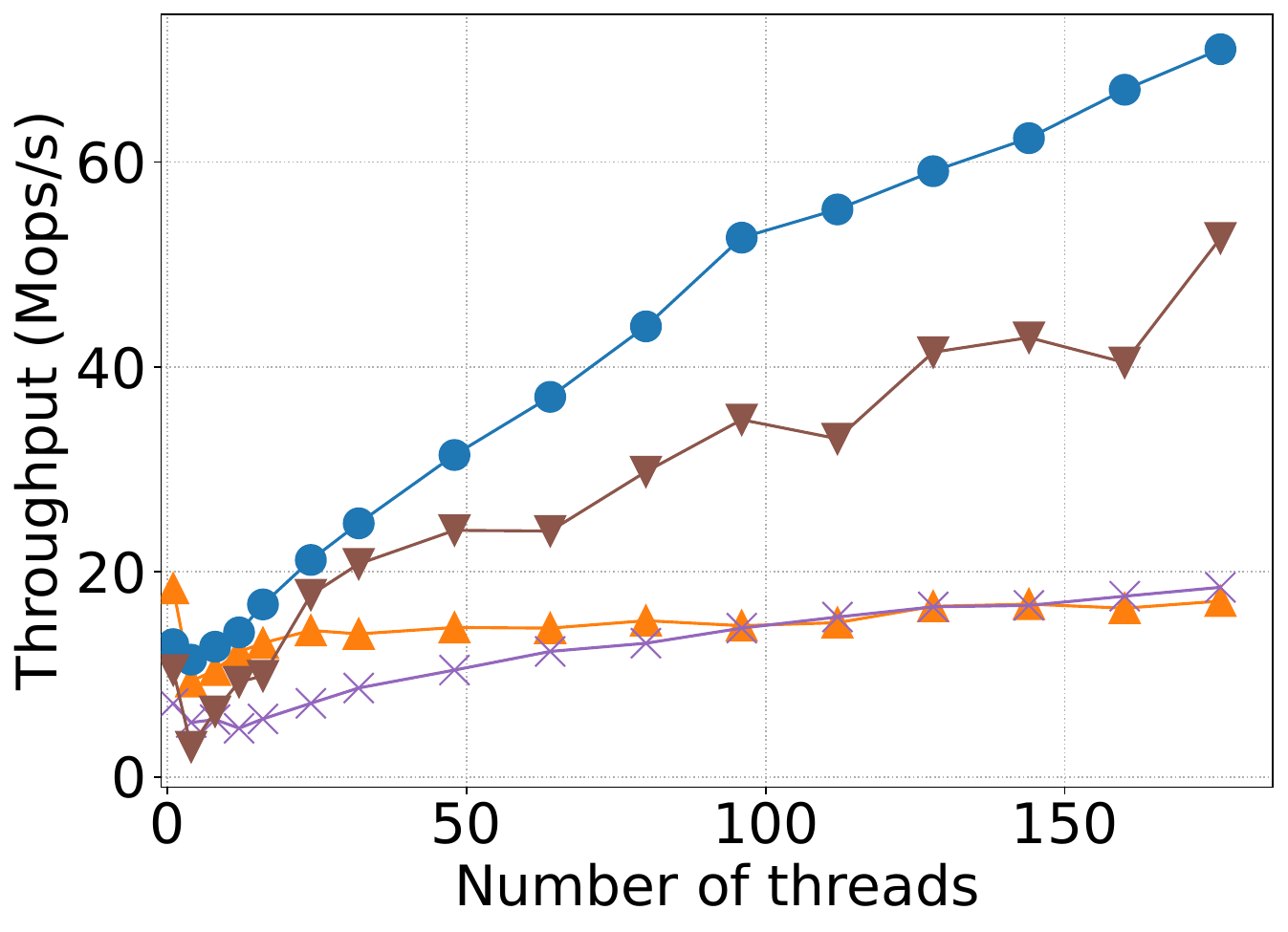}    
    \subcaption{90\% \FetchAndAdd{}, 32 cycles, throughput}
    \label{fig:count_lesswork}
    \end{minipage}
    \\
    \vspace{6pt}
    \begin{minipage}{0.32\textwidth}
    \includegraphics[width=\textwidth]{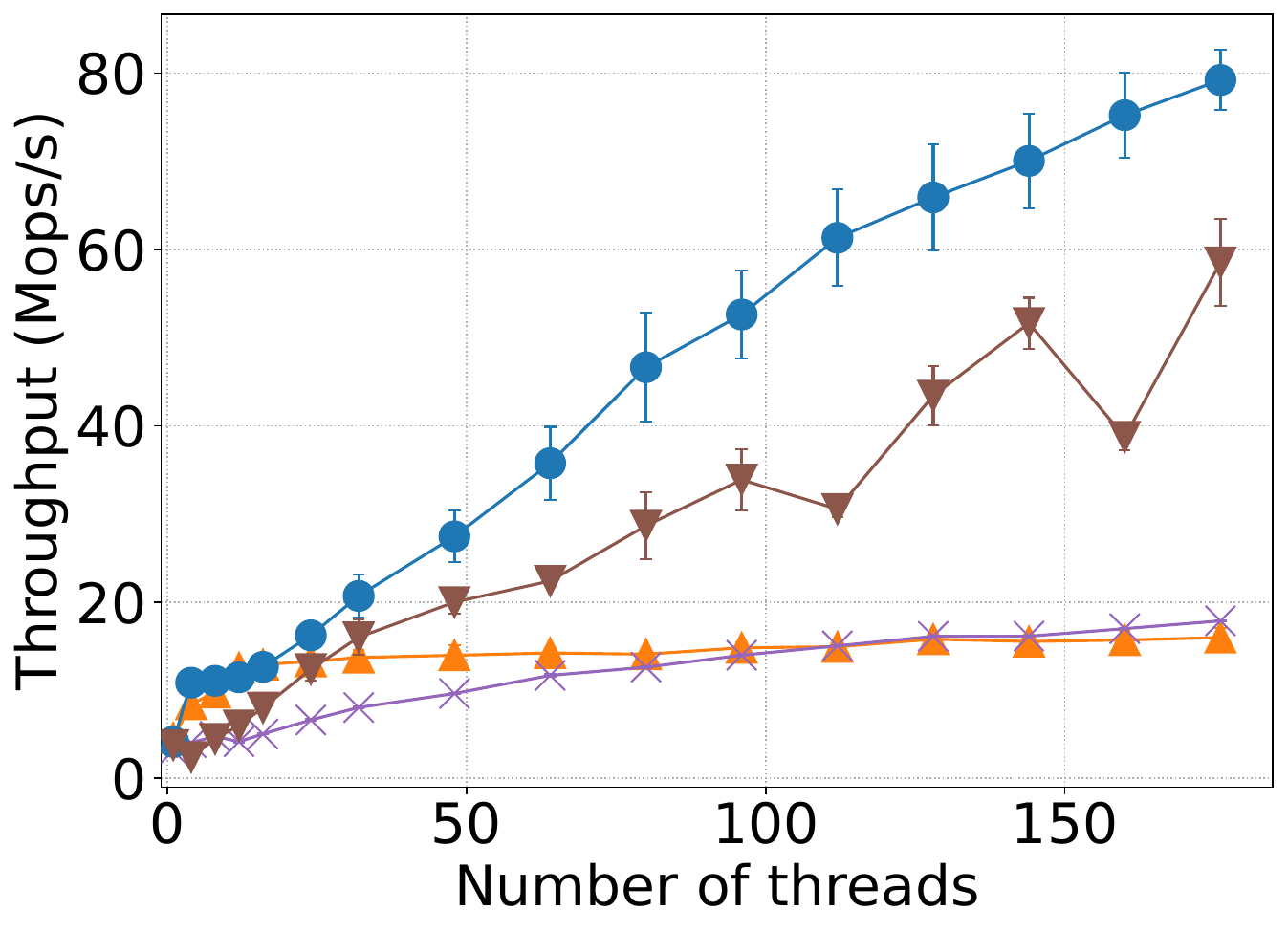}
    \subcaption{100\% \FetchAndAdd{}, 512 cycles, throughput}
    \label{fig:count_100}
    \end{minipage}
    \begin{minipage}{0.32\textwidth}
    \includegraphics[width=\textwidth]{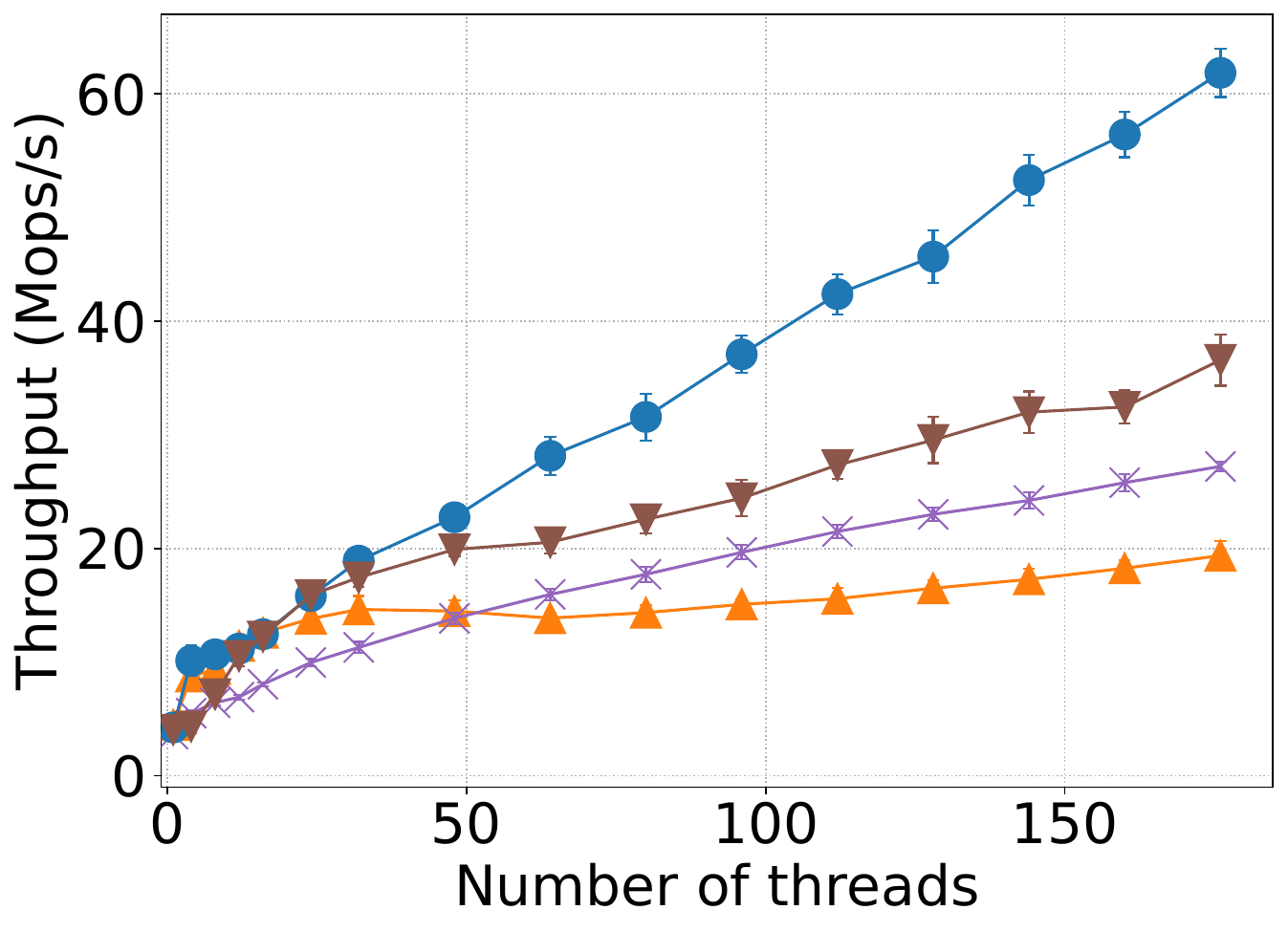}
    \subcaption{50\% \FetchAndAdd{}, 512 cycles, throughput}
    \label{fig:count_50}
    \end{minipage}
    \begin{minipage}{0.32\textwidth}
    \includegraphics[width=\textwidth]{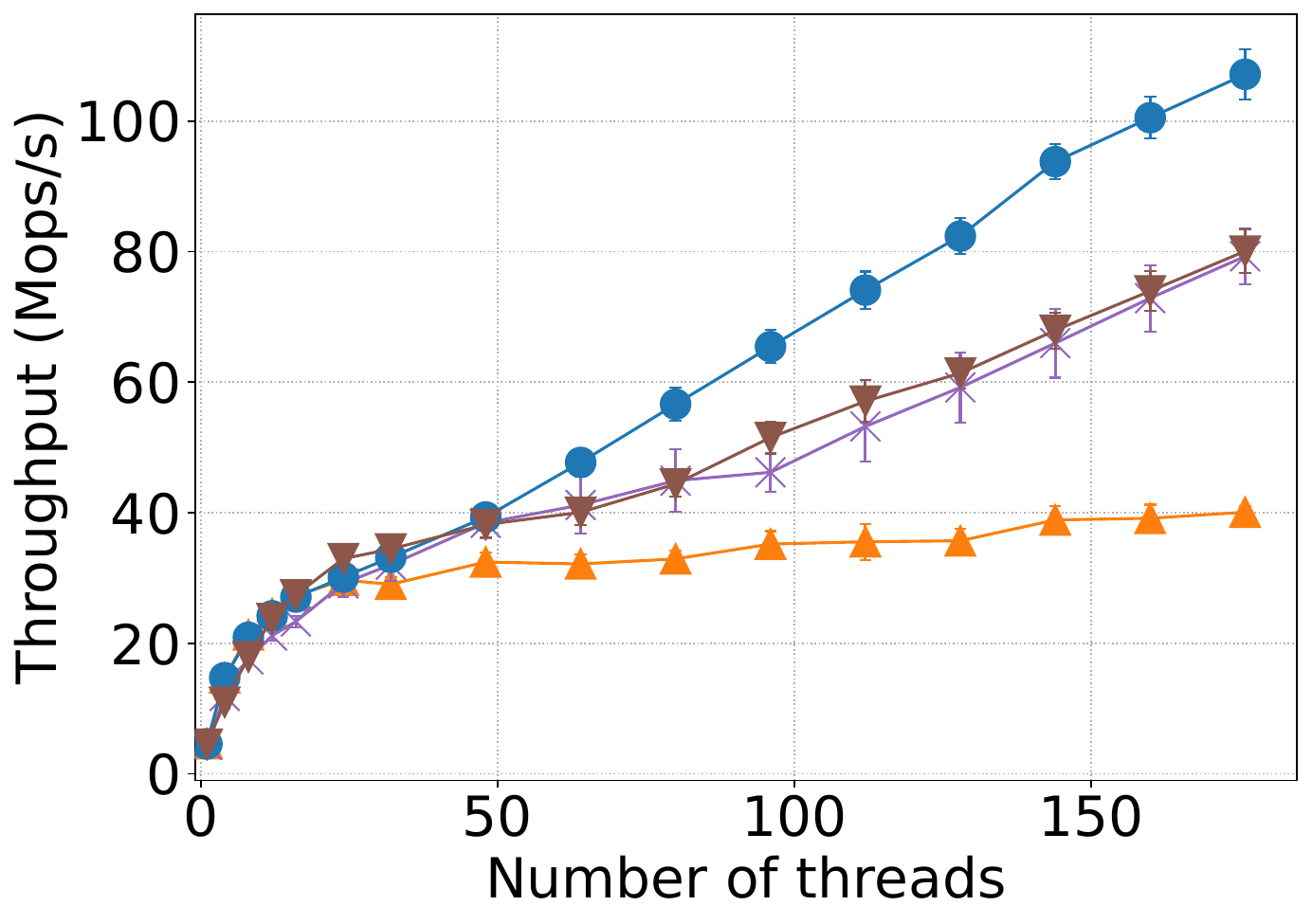}
    \subcaption{10\% \FetchAndAdd{}, 512 cycles, throughput}
    \label{fig:count_10}
    \end{minipage}
    \caption{Comparing throughput and fairness of Aggregating Funnels, Combining Funnels, and hardware \FAA. 
    }
    \label{fig:count_main}
\end{figure*}

In this section, we compare the performance of our algorithm with Combining Funnels \cite{SZ00} and hardware \FAA.
We tested the Combining Funnels by varying the depth and width of the funnel, and found that the best performing variant uses $\left\lceil \log(p) \right\rceil -1$ levels, halving the width at every level.
For Aggregating Funnels, we use 6 \Aggs\ and distribute threads evenly as mentioned above.
For recursive Aggregating Funnels (described in \Cref{recursive-sec}), we use the best performing variant which uses $m = \lceil p/6 \rceil$ \Aggs\ for the fetch-and-add object $O$,
and replaces the \x{Main} variable of $O$ by another instance of our algorithm with $m' = 6$ \Aggs, with threads distributed evenly.

\Cref{fig:count_main} shows Aggregating Funnels are faster than Combining Funnels in all cases, and outperform hardware \FAA\ after 30 threads.
Aggregating Funnels scale the best in all experiments, and Aggregating Funnels are up to 4x faster than both Combining Funnels and hardware \FAA\ for high thread counts.

For low thread counts, Combining Funnels have lower throughput than other algorithms, but they scale better than hardware \FAA\ and slightly outperform hardware \FAA\ with more threads in \Cref{fig:count_thr}.
Recursive Aggregating Funnels are expected 
to scale better than single-level Aggregating Funnels as $p$ gets very large since it reduces contention further, but it did not achieve better throughput when testing it with up to 176 threads.
As discussed in \Cref{allocation-scheme}, having fewer writing threads on the \x{Main} variable is advantageous in read-heavy workloads. This effect can be seen by comparing \Cref{fig:count_50} and \Cref{fig:count_10}. Since recursive Aggregating Funnels and Combining Funnels have fewer writing threads on the \x{Main} variable than the Aggregating Funnels, their throughput increases more as the workload has more read operations.

Varying the additional work did not significantly affect the throughput curve.
Comparing \Cref{fig:count_thr} and \Cref{fig:count_lesswork}, we see that only results with fewer than 8 threads were affected, and differences are negligible for higher thread counts.

Aggregating Funnels have higher fairness compared to hardware \FAA\ for 32 or more threads, as shown in \Cref{fig:count_fair}.
Previous work suggests the reason hardware \FAA\ becomes unfair at high contention is that some threads benefit from getting exclusive access to the variable's cache line for longer~\cite{BSB19}.
Aggregating Funnels, however, mitigate this unfairness with three changes.
In both \Agg{}'s \x{value} and \x{Main} variable, the maximum number of contending threads is smaller. This allows each cache line to be used more fairly across contending threads.
Furthermore, a delegate thread  with fast \FAA\ access to \x{Main} also benefits the other threads in the same \Agg{}.
Notably, Combining Funnels have high fairness, due to the wider and deeper funnel configuration, and assigning random locations for each operation. 

We also ran the same experiments in \Cref{fig:count_main} on AMD EPYC 9B14 processors, as well as 1st, 3rd and 5th Gen Intel Xeon processors.
Hardware \FAA\  performed differently on the different processors. 
In contrast, Aggregating Funnels scaled similarly in all machines and workloads we tested.
On our primary machine (with 4th Gen Intel Xeon processors), hardware \FAA\  stopped scaling after 30 threads, plateauing around 18Mops/s
(\Cref{fig:count_thr}).
On the newer 5th Gen Intel Xeon processor, hardware \FAA\  plateaued at around 20Mop/s.
In older Intel machines, hardware \FAA\  scaled better than our primary machine, plateauing around 30Mops/s.
In the AMD machine, hardware \FAA\  scaled well on one socket but its throughput sharply dropped when moving to 2 sockets, plateauing around 40Mops/s.
Across all the machines that we tested, Aggregating Funnels outperformed hardware \FAA\  at high thread counts.

\subsection{\op{Fetch\&AddDirect} for High-Priority Threads}
\label{sec:high-priority}

\begin{figure*}
    \begin{minipage}{\textwidth}
    \includegraphics[width=0.9\textwidth]{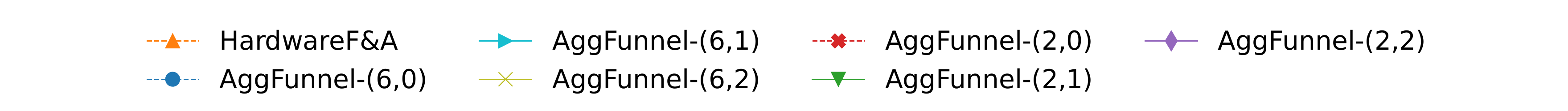}
    \end{minipage}
    \\
    \begin{minipage}{0.32\textwidth}
    \includegraphics[width=\textwidth]{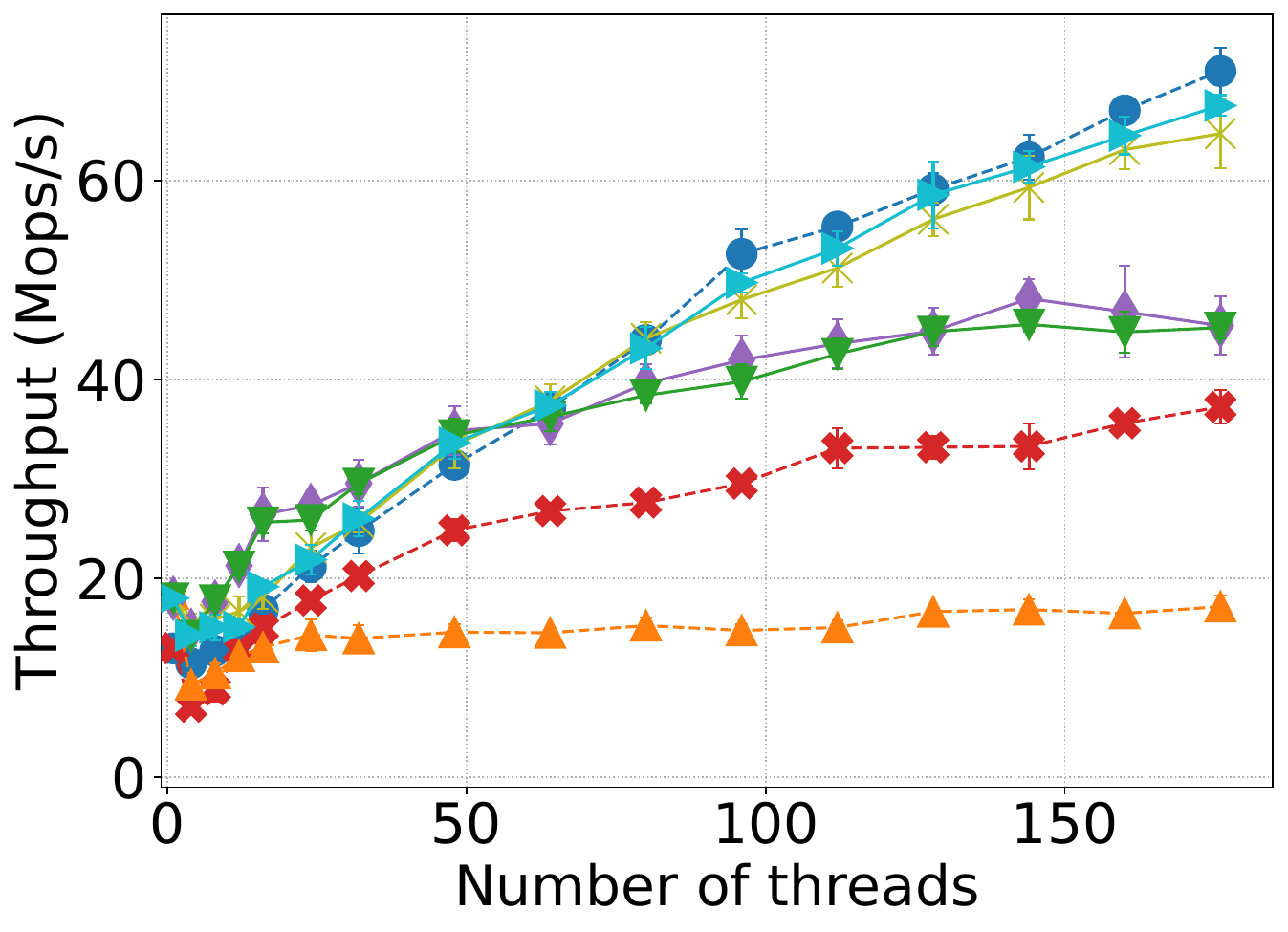}
    \subcaption{Throughput}
    \label{fig:skew_thr}
    \end{minipage}
    \begin{minipage}{0.32\textwidth}
    \includegraphics[width=\textwidth]{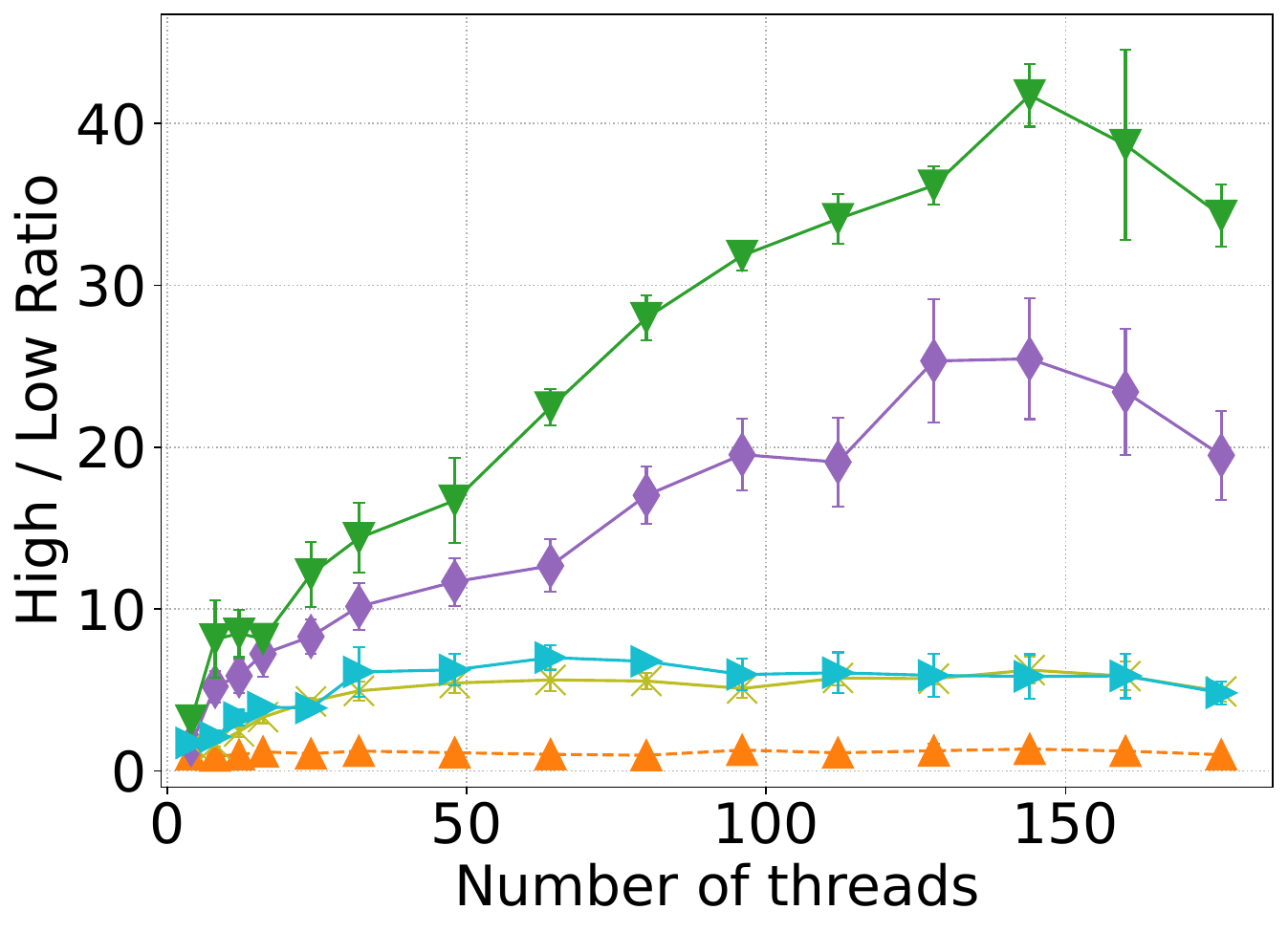}
    \subcaption{Average throughput ratio between high- and low-priority threads}
    \label{fig:skew_fair}
    \end{minipage}
    \begin{minipage}{0.32\textwidth}    
    \includegraphics[width=\textwidth]{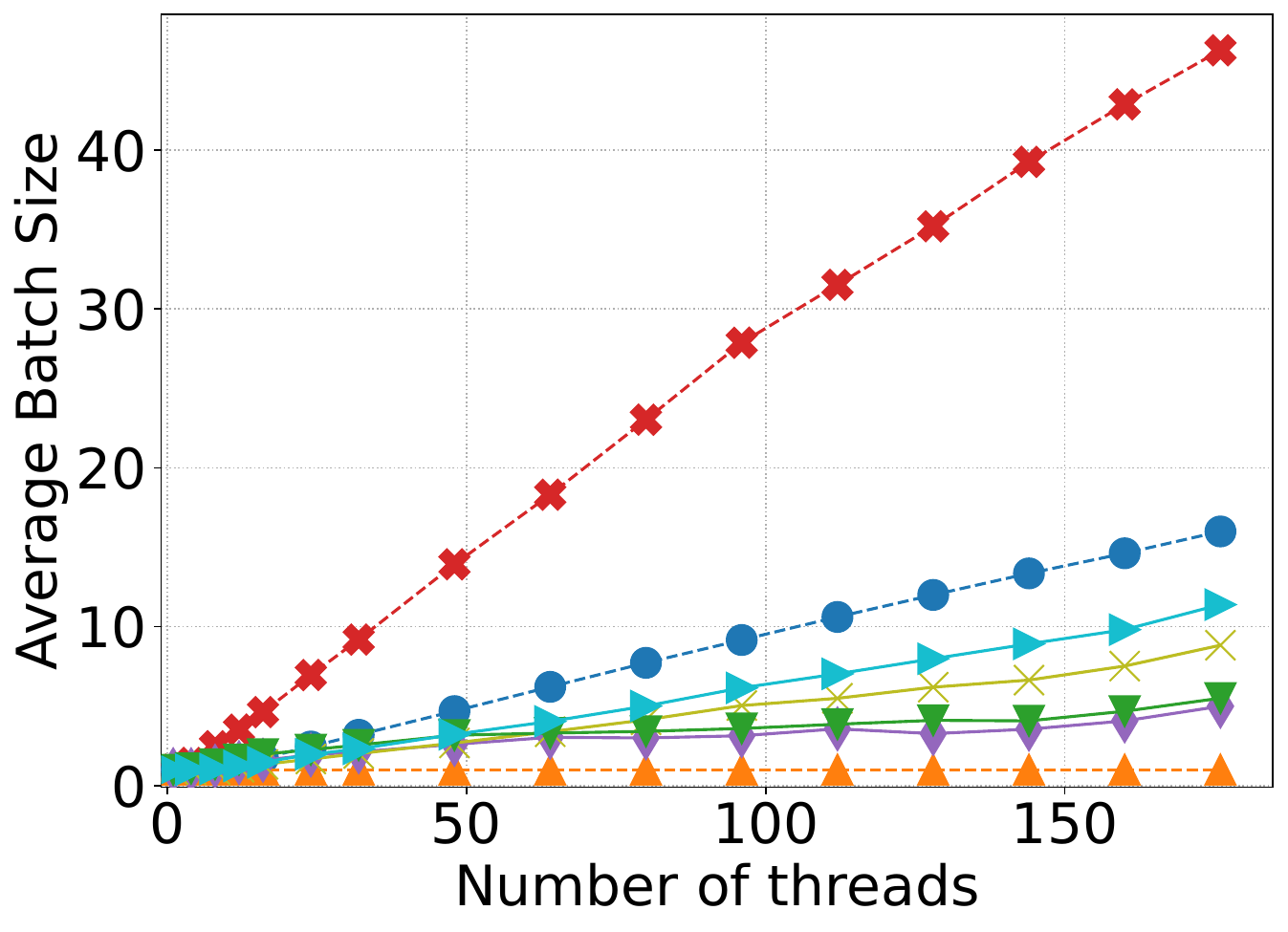}    
    \subcaption{Average batch size}
    \label{fig:skew_batch}
    \end{minipage}
    \caption{
    \FetchAndAdd\ performance with high-priority threads. 90\% \FetchAndAdd{}, 32 cycles of additional work.}
    \label{fig:skew_main}
\end{figure*}

\begin{figure*}
    \begin{minipage}{\textwidth}
    \includegraphics[width=0.9\textwidth]{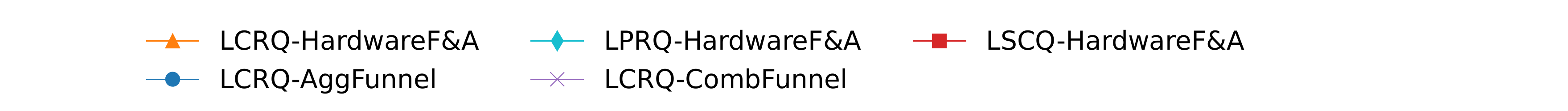}
    \end{minipage}
    \\
    \begin{minipage}{0.32\textwidth}
    \includegraphics[width=\textwidth]{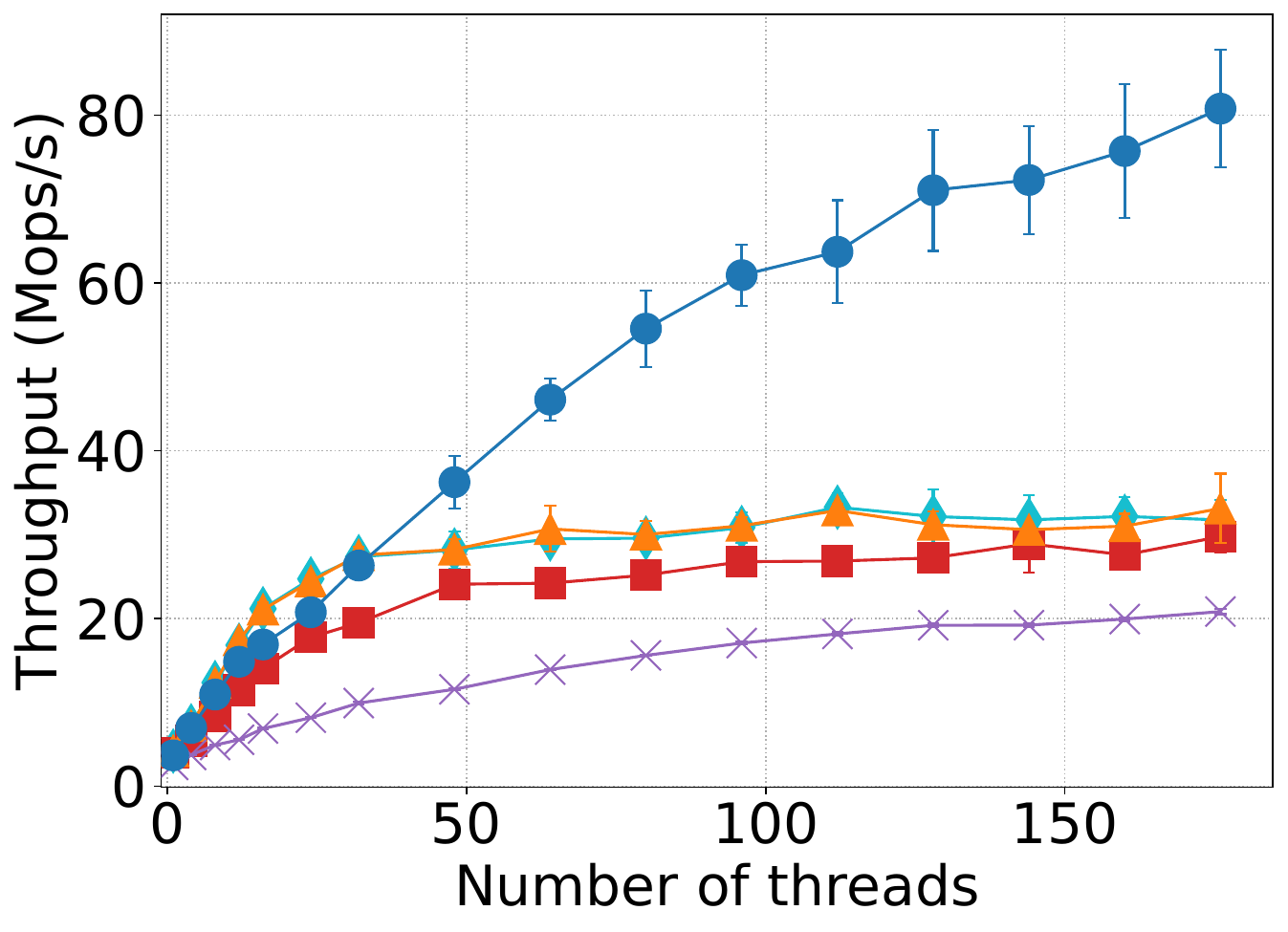}
    \subcaption{Pairwise enq-deq, empty initial queue}
    \label{fig:queue_default}
    \end{minipage}
    \begin{minipage}{0.32\textwidth}
    \includegraphics[width=\textwidth]{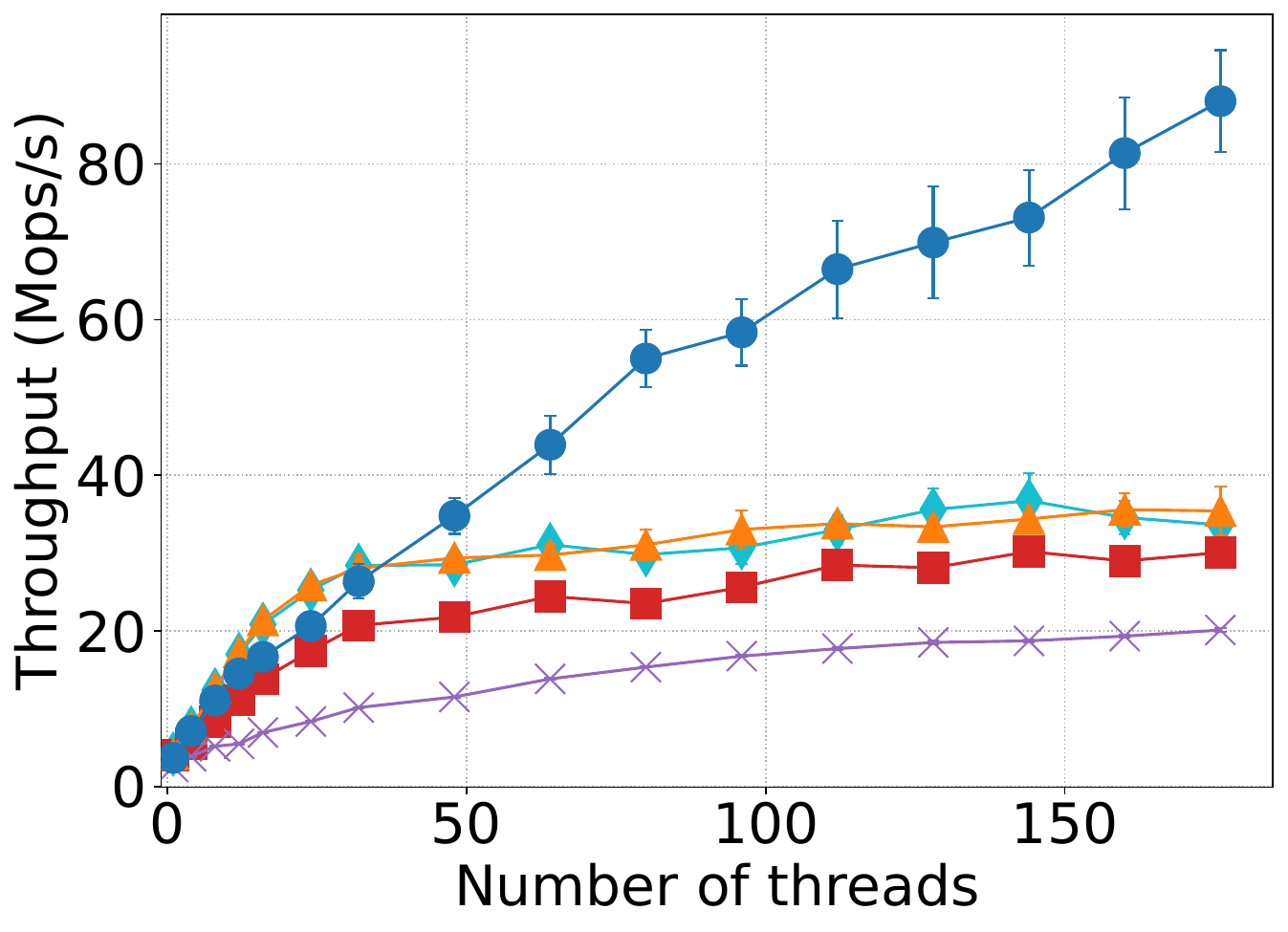}
    \subcaption{Pairwise 4enq-4deq, empty initial queue}
    \label{fig:queue_consec4}
    \end{minipage}
    \begin{minipage}{0.32\textwidth}    
    \includegraphics[width=\textwidth]{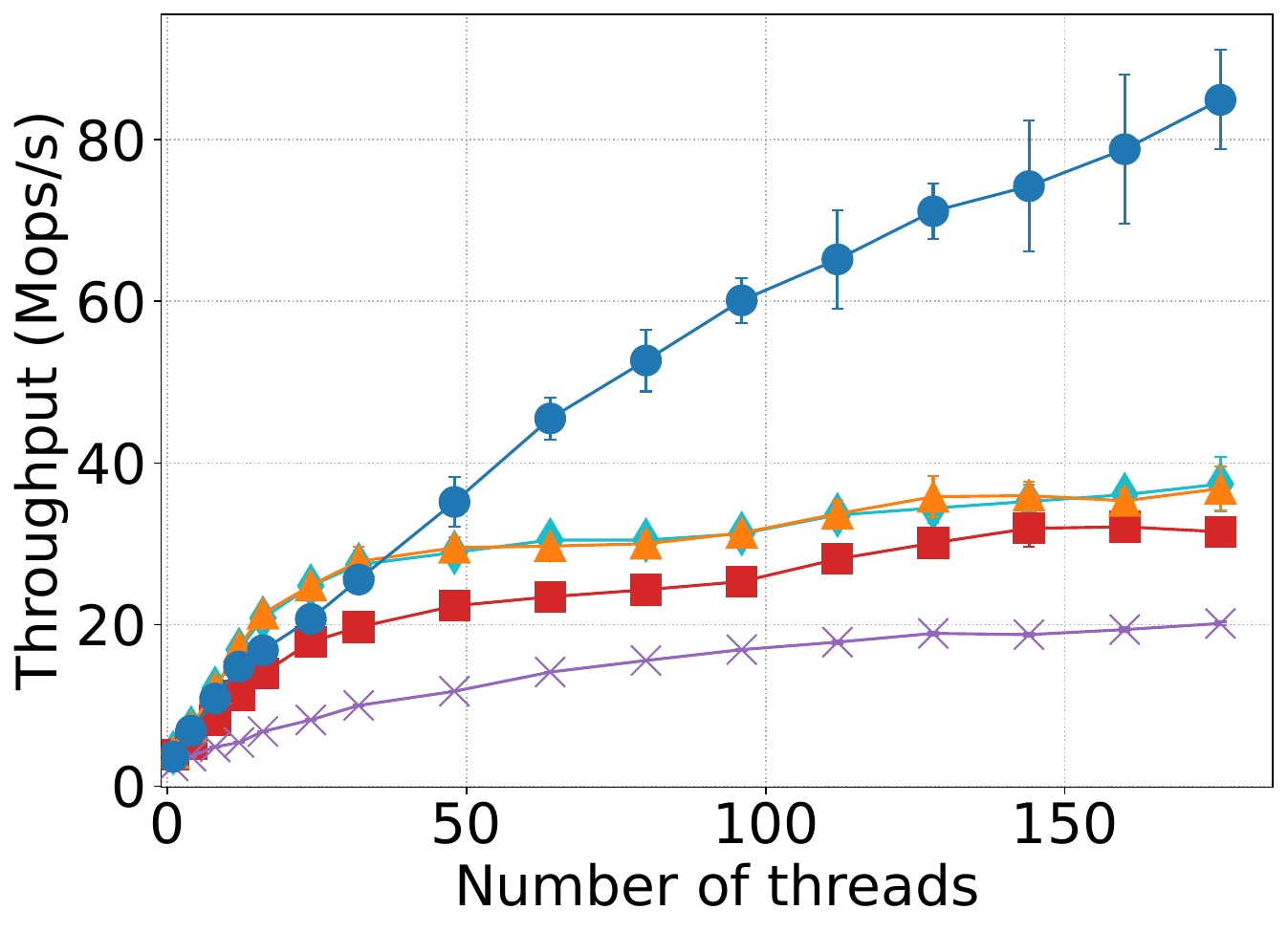}    
    \subcaption{Pairwise enq-deq, initial queue size 500}
    \label{fig:queue_lessinit}
    \end{minipage}
    \caption{Queue performance using different fetch-and-add implementations in LCRQ.}
    \label{fig:queue_main}
\end{figure*}

As mentioned in line \ref{direct} of \Cref{code-overflow}, our implementations support \op{Fetch\&AddDirect}, which performs a \FAA\ directly on \x{Main} and therefore has lower expected latency.
This characteristic can be utilized as an asset when different levels of priority are desired.
For example, a program may prioritize a specific thread's progress over different threads by calling \op{Fetch\&AddDirect}, when it is going through a critical section that stalls the other threads.
Any thread can decide when to use \op{Fetch\&AddDirect} at runtime.

In this section, we experiment with an asymmetric allocation scheme \op{AggFunnel}--$(m,d)$, where \x{d} threads are  \emph{high-priority threads} that call \op{Fetch\&AddDirect}, and the other $p-d$ low-priority threads will start from $m$ \x{Positive} \Aggs\ evenly, as explained in \Cref{allocation-scheme}.
For \Cref{fig:skew_main}, we ran schemes with $m=2,6$ and $d=0,1,2$, and only 32 cycles of additional work to highlight the findings.

\Cref{fig:skew_thr} shows throughput for different parameters.
With $m=6$, the total throughput was not significantly affected by having high-priority threads.
However, with $m=2$, the total throughput increased when high-priority threads were present.
This effect was more visible with less additional work.
We believe this is because high-priority threads can do consecutive \FetchAndAdd\ operations on \x{Main} variable, which can significantly decrease the number of cache loads.

\Cref{fig:skew_fair} shows that average throughput of high-priority threads is up to 40x higher than that of low-priority threads, while the total throughput across all threads is higher than or similar to that of symmetric allocation scheme.
\Cref{fig:skew_batch} also confirms that high-priority threads write to the \x{Main} variable more often than low-priority threads, decreasing the average batch size. 
(One \op{Fetch\&AddDirect} operation counts as one batch.)
These results show that a few high-priority threads can be introduced to reduce latency for performance critical code without sacrificing overall throughput.

\subsection{Queue Benchmark}
\label{sec:queue-benchmark}

Since our \FetchAndAdd\ algorithm supports all hardware  primitives, we can easily replace a hardware \FAA\ object in various applications to mitigate the contention bottleneck.
As mentioned in \Cref{related-queue}, one significant application of \FetchAndAdd\ is in concurrent queues.
To confirm the usability of Aggregating Funnels, we ran a concurrent queue benchmark, with existing queues (LCRQ \cite{MA13}, LSCQ \cite{R19}, and LPRQ \cite{RK23}) with hardware \FAA{}, and LCRQ with Aggregating Funnels and Combining Funnels.

We modified the previously published artifact \cite{RK23} with our implementations of Aggregating Funnels.
We ran the benchmark with the existing docker configuration in the artifact, which uses \texttt{clang++-13} and \texttt{jemalloc}, with the \texttt{numactl -i all} command to distribute memory evenly across the sockets. 
Similar to the \FetchAndAdd\ benchmarks, we added an average of 512 cycles of work between successive enqueues and dequeues by the same thread.
\Cref{fig:queue_main} shows total throughput, which is double the transfer rate reported in~\cite{RK23}.

\Cref{fig:queue_main} illustrates that simply replacing hardware \FAA\ with the more scalable Aggregating Funnels \FetchAndAdd\ achieves much higher throughput.
In all three scenarios shown in the figure,
LCRQ with Aggregating Funnels has up to 2.5x higher throughput than LCRQ with hardware \FAA, and more than 3.5x higher throguhput than LCRQ with Combining Funnels for high thread counts.
        \section{Conclusion and Future Work}
\label{sec:conclusion}

In this paper, we designed the Aggregating Funnel algorithm for fetch-and-adds.
Our microbenchmarks show that Aggregating Funnels are very effective at dissipating contention: outperforming hardware fetch-and-add and the state-of-the-art Combining Funnels algorithm over a variety of workloads.
We demonstrated that the speed-ups observed in the microbenchmarks translate to higher-level applications by deploying our Aggregating Funnels in LCRQ.  
Replacing the hardware fetch-and-add objects with our Aggregating Funnels yields a significant (up to $2.5$x) speed-up in the performance of this state-of-the-art concurrent queue.

This work opens up many interesting avenues for future exploration, including:
(1) {\em Adapting the algorithm to new settings.}
For example, exploring non-blocking variants, NUMA-awareness, direct implementation in hardware, adaptive assignment of processes to \Aggs, and incorporating elimination \cite{ST97} to speed-up the common cases where increments and decrements are only by one. 
(2) {\em Deploying Aggregating Funnels in fetch-and-add applications beyond LCRQ.}
For example, the camera object in~\cite{WBBFR021}, the sequence number mechanism in~\cite{FPR19},
improving the performance of timestamping in software transactional memory algorithms such as TL-II~\cite{DSS06}, 
and more generally, in concurrent timestamping in database transactions and other database applications~\cite{S84}.

\begin{acks}
	We thank the anonymous reviewers of the paper and artifact for their comments. 
   This work was supported by 
   the MIT Undergraduate Research Opportunities Program and Ralph L.\ Evans (1948) Endowment Fund;
   National Science Foundation grants CCF-1845763, CCF-2316235, and CCF-2403237;
   a Google Faculty Research Award and Research Scholar Award;  
   the Natural Sciences and Engineering Research Council of Canada;
   the Hellenic Foundation for Research and Innovation under the Second Call for Research Projects to support Faculty Members and Researchers (Project: PERSIST, number: 3684); and the Greek Ministry of Education, Religious Affairs and Sports call SUB 1.1 -- Research Excellence Partnerships (Project: HARSH, code: $\Upsilon\Pi$ 3TA-0560901)
\end{acks}

 \newpage
	\bibliographystyle{ACM-Reference-Format}
	\bibliography{strings,biblio}
\appendix
\ 
\section{Artifact Evaluation Appendix}
\label{sec:ae}

\subsection{Abstract}

This artifact contains the source code and scripts to reproduce all the graphs in Section \ref{sec:experiments}.
For an up-to-date version of the Aggregating Funnels library, please visit our repository on GitHub: \url{https://github.com/Diuven/aggregating-funnels/tree/artifact-submission}.

\subsection{Artifact check-list (meta-information)}

\begin{itemize}
	\item \textbf{Algorithm:} The Aggregating Funnels and recursive Aggregating Funnels algorithm described in \Cref{sec:algorithms}.
	\item \textbf{Program:} microbenchmarks
	\item \textbf{Compilation:} \texttt{g++13}, \texttt{clang-13}
	\item \textbf{Run-time environment:} Ubuntu 24.04 LTS
	\item \textbf{Hardware:} Multi-core machine, preferably with Intel 4th Gen Xeon or newer processor with at least 64 logical cores
	\item \textbf{Output:} Graphs from Section \ref{sec:algorithms} as png files.
	\item \textbf{Experiments workflow:} One script for compiling, running, and generating graphs for \FAA\ benchmarks, and one script for compiling, running, and generating graphs queue benchmarks.
    the experiments and one script for generating all the graphs.
	\item \textbf{Disk space required (approximately):} 8 GB
	\item \textbf{Time needed to prepare workflow:} approximately 15 minutes
	\item \textbf{Time needed to complete experiments:} approximately 6 hours
	\item \textbf{Publicly available:} yes
	\item \textbf{Code licenses:} MIT License
\end{itemize}

\subsection{Description}

\subsubsection{How delivered}

The artifact is available on Zenodo \url{https://zenodo.org/records/14602039}.

\subsubsection{Hardware dependencies}

To accurately reproduce our experimental results, a multi-core machine with Intel 4th Gen Xeon or newer processor with at least 64 logical cores is recommended.

\subsubsection{Software dependencies}

Our artifact is expected to run correctly under a variety of Linux x86\_64 distributions.
\texttt{numactl} is needed to evenly distribute the memory allocations across multiple sockets.
All other dependencies are included in the docker configuration, therefore only docker runtime supporting x84\_64 Ubuntu 24.04 is required.

\subsubsection{Data sets}

None.

\subsection{Installation}

For the detailed and updated instruction, please refer to the README file of our Github repository. \url{https://github.com/Diuven/aggregating-funnels/tree/artifact-submission}

\begin{enumerate}    
    \item Build the docker image (install docker if you haven't)
    \texttt{docker build --network=host --platform linux/amd64 -t aggfunnel .}
    
    \item Launch the docker container as an interactive shell. This command also complies all the necessary binaries. Remaining commands should be run inside the docker container.
    \texttt{docker run -v .:/home/ubuntu/project -it --privileged --network=host aggfunnel}
    Note: This command mounts the current directory aggregating-funnels/ into the docker container, so both are synchronized.
\end{enumerate}

\subsection{Experiment workflow}

After compiling, run \texttt{./scripts/run\_counter\_bench.sh \&\& ./scripts/run\_queue\_bench.sh} inside the docker to run and generate all the graphs.

\subsection{Evaluation and expected results}

On a machine with ~128 logical cores and with recently released processor, the throughput of Aggregating Funnels should be very similar to those reported in this paper. Note that hardware \FAA\ may perform differently on different processors, but the Aggregating Funnels should scale better than hardware \FAA\ at high threads, as discussed at the end of \Cref{sec:main_benchmark}.

\subsection{Experiment customization}

For instructions on how to customize the number of threads, workload, and the allocation scheme in each experiment, please see the README file included in the artifact.

\subsection{Notes}

None.

\subsection{Methodology}
Submission, reviewing and badging methodology:
\begin{itemize}
	\item \url{https://ctuning.org/ae/submission-20190109.html}
	\item \url{https://ctuning.org/ae/reviewing-20190109.html}
	\item \url{https://www.acm.org/publications/policies/artifact-review-badging}
\end{itemize}
\end{document}